\documentclass[11pt]{article}

\usepackage{fullpage,authblk}

\usepackage{titlesec}

\usepackage{ifthen}
\usepackage[funcfont=italic,full]{./complexity/complexity}
\usepackage{subfig}
\usepackage{hyperref}
\usepackage{graphicx}
\usepackage{amssymb}
\usepackage{amsfonts}
\usepackage{lipsum}
\usepackage{epstopdf}
\usepackage{amsopn}
\usepackage{braket}
\usepackage{algorithm}     
\usepackage{algpseudocode} 
\usepackage{enumitem}

\renewcommand{\Re}{\mathbb{R}}
\newcommand{\Q}{\mathbb{Q}}
\newcommand{\Z}{\mathbb{Z}}
\renewcommand{\P}{\mathcal{P}}
\renewcommand{\S}{\mathcal{S}}
\renewcommand{\C}{\mathcal{C}}
\renewcommand{\D}{\mathcal{C}^\circ}
\renewcommand{\E}{\mathcal{E}}
\renewcommand{\K}{\mathcal{K}}
\newcommand{\Qq}{\mathcal{Q}}
\newcommand{\bigmid}{\;\middle|\;}

\DeclareMathOperator{\conv}{conv}
\DeclareMathOperator{\inter}{int}
\DeclareMathOperator{\cone}{cone}
\DeclareMathOperator{\argmin}{argmin}
\DeclareMathOperator{\argmax}{argmax}

\DeclareMathOperator{\st}  {s.\!t.}

\algdef{SE}[DOWHILE]{Do}{doWhile}{\algorithmicdo}[1]{\algorithmicwhile\ #1}

\makeatletter
\newenvironment{mysubequations}[1]
 {%
  \addtocounter{equation}{-1}%
  \begin{subequations}
  \renewcommand{\theparentequation}{#1}%
  \def\@currentlabel{#1}%
 }
 {%
  \end{subequations}\ignorespacesafterend
 }
\makeatother

\usepackage[authoryear]{natbib}
\newcommand{\mycite}[2]{\citet{#2}}
\newcommand{\mycitep}[1]{\citep{#1}}
\usepackage{amsmath}
\usepackage{amsthm}
\newcommand{\myparagraph}[1]{\paragraph{#1.}}

\newcommand{\Sectiref}[1]{Section~\ref{#1}}

\newcommand{\Figurref}[1]{Figure~\ref{#1}}
\newcommand{\Theorref}[1]{Theorem~\ref{#1}}

\newtheorem{theorem}{Theorem}
\newtheorem{proposition}{Proposition}
\newtheorem{definition}{Definition}

\newtheorem{lemma}{Lemma}

\let\nativeeqref\eqref

\graphicspath{{fig/}}

\begin{document}

\title{On the Complexity of Inverse Mixed Integer Linear Optimization}
\author{Aykut Bulut\thanks{E-mail: \texttt{aykutblt@gmail.com}}}
\author{Ted K. Ralphs\thanks{E-mail: \texttt{ted@lehigh.edu}}}
\affil{Department of Industrial and Systems Engineering, Lehigh University, USA}

\maketitle

\begin{abstract}
Inverse optimization is the problem of determining the values of missing input
parameters for an associated \emph{forward problem} that are closest to given
estimates and that will make a given \emph{target vector} optimal. This study
is concerned with the relationship of a particular inverse mixed integer
linear optimization problem (MILP) to both the forward problem and the
separation problem associated with its feasible region.
We show that a decision version of the inverse MILP in which a primal bound is
verified is $\coNP$--complete, whereas primal bound verification for the
associated forward problem is $\NPcomplexity$--complete, and that the optimal
value verification problems for both the inverse problem \emph{and} the
associated forward problem are complete for the complexity class
$\Dcomplexity^{\Pcomplexity}$. We also describe a cutting-plane algorithm for
solving inverse MILPs that illustrates the close relationship between the
separation problem for the convex hull of solutions to a given MILP and the
associated inverse problem. The inverse problem is shown to be equivalent to
the separation problem for the radial cone defined by all inequalities that
are both valid for the convex hull of solutions to the forward problem and
binding at the target vector. Thus, the inverse, forward, and separation
problems can be said to be equivalent.

\end{abstract}

\section{Introduction}
\label{sec:introduction}

In this paper, we study the relationship of the inverse integer linear
optimization problem to both the optimization problem from which it arose,
which we refer to as the \emph{forward problem}, and the associated separation
problem. We show that these three problems have a strong relationship from an
algorithmic standpoint by describing a cutting-plane algorithm for the inverse
problem that uses the forward problem as an oracle and also solves the
separation problem. 
From a complexity standpoint, we show that certain decision versions of these
three problems are all complete for the complexity class
$\Dcomplexity^{\Pcomplexity}$, introduced originally
by~\mycite{Papadimitriou}{Papadimitriou:1982:CF:800070.802199}. Motivated by
this analysis, we argue that the optimal value verification problem is a more
natural decision problem to associate with many optimization problems and that
$\Dcomplexity^{\Pcomplexity}$ may provide a more appropriate class in which to
place difficult discrete optimization problems than the more commonly cited
$\NPcomplexity$-hard.

An \emph{optimization problem} is that of determining a member of a feasible
set (an \emph{optimal solution}) that minimizes the value of a given objective
function. The feasible set is typically described as the points in a vector
space satisfying a given set of equations, inequalities, and disjunctions
(usually in the form of a requirement that the value of a certain element of
the solution take on values in a discrete set).

An \emph{inverse optimization problem}, in contrast, is a related problem in
which the description of the original forward optimization problem is
incomplete (some parameters are missing or cannot be observed), but a full or
partial solution \emph{can} be observed. The goal is to determine values for
the missing parameters with respect to which the given solution would be
optimal for the resulting complete problem. Estimates for the missing
parameters may be given, in which case the goal is to produce a set of
parameters that is as close to the given estimates as possible by some metric.

The forward optimization problem of interest in this paper is the
\emph{mixed integer linear optimization problem}
\begin{equation}
\max_{x \in \S} d^\top x,
\label{eq:milp} \tag{MILP}
\end{equation}
where $d \in \Q^n$ and
\begin{equation*}
\label{eq:set-s}
\S = \left\{x \in \Re ^n \bigmid Ax \leq b\right\} \cap
(\Z ^r \times \Re^{n-r})
\end{equation*}
for $A \in \Q^{m \times n}$, $b \in \Q^m$ and for some nonnegative integer
$r$. In the case when $r = 0$, \eqref{eq:milp} is known simply as
a \emph{linear optimization problem} (LP). 

One can associate a number of different inverse problems with~\eqref{eq:milp},
depending on what parts of the description $(A, b, d)$ are unknown and what
form the objective function of the inverse problem takes. Here, we study the
case in which the objective function $d$ of the forward problem is the unknown
element of the input, but in which $A$ and $b$, along with
a \emph{target vector} $x^0 \in \Q^n$, are given. A feasible solution to the
inverse problem (which we refer to as a \emph{feasible objective}) is any
$\hat{d} \in \Re^n$ for which $\hat{d}^\top x^0 = \max_{x \in \S} \hat{d}^\top
x$.

It is important in the analysis that follows to be precise about the
assumptions on the target vector $x^0$.
Our initial informal statement of the problem implicitly assumed that
$x^0 \in \S$, since otherwise, $x^0$ cannot technically be an optimal
solution, regardless of the objective function chosen. On the other hand,
neither the more precise mathematical definition given in the preceding
paragraph nor the mathematical formulations we introduce shortly require
$x^0 \in \S$ and both can be interpreted even when $x^0 \not\in \S$. As a
practical matter when solving inverse problems in practice, this subtle
distinction is usually not very important, since membership in $\S$ can be
verified in a preprocessing step if necessary. However, in the context of
complexity analysis and in considering the relationship of the inverse problem
to the related separation and optimization problems, this point \emph{is}
important and we return to it. For example, if we do not make this assumption,
the inverse optimization problem can be seen to be equivalent to both the
separation problem and the problem of verifying a given dual bound on the
optimal solution value of an MILP.

For these and other reasons that will become clear, we do not assume
$x^0 \in \S$, but this may make some aspects of what follows a bit ambiguous.
To resolve any ambiguity, we replace $\S$ with the augmented set $\S^+
= \S \cup \{x^0\}$, when appropriate, in the remainder of the paper.

\subsection{Formulations}

We now present several mathematical formulations of what we refer to from now
on as the \emph{inverse mixed integer linear optimization problem} (IMILP). A
straightforward formulation of this problem that explains why we refer to the
general class of problems as ``inverse'' problems is as that of computing the
mathematical inverse of a function that is parametric in some part of the
input to a given problem instance. In this case, the relevant function
is
\begin{equation*}
\phi(d) = \argmax_{x \in \S^+} d^\top x.
\end{equation*}
In terms of the function $\phi$, a feasible objective is any element of
the preimage $\phi^{-1}(x^0)$. To make the IMILP an optimization
problem in itself, we add an objective function $f:\Q^n \rightarrow \Q$, to
obtain the general formulation
\begin{equation}
\min_{d \in \phi^{-1}(x^0)} f(d).  \label{eq:inv} \tag{INV}
\end{equation}
The traditional objective function used for inverse problems in the literature
is $f(d) := \|c - d\|$, the minimum norm distance from $d$ to a given
\emph{estimated objective function} $c \in \Q^n$ (the specific norm is not
important for defining the problem, but we assume a $p$-norm when proving the
formal results). This choice of objective, although standard, has some
nonintuitive properties. First, it is not scale-invariant---scaling a given
feasible objective $\hat{d}$ changes the resulting objective function value.
In other words, $f(\hat{d}) \not= f(\lambda \hat{d})$ for $\lambda \not= 1$,
so $\hat{d}$ and $\lambda \hat{d}$ do not have the same objective function
value, although it is clear that $\hat{d}$ and $\lambda \hat{d}$ are
equivalent solutions in most settings. Second, the objective function value is
always nonnegative. Both these properties have implications we discuss further
below. 

The formulation~\eqref{eq:inv} does not suggest any direct connection to
existing methodology for solving mathematical optimization problems, so we
next discuss several alternative formulations of the problem as a standard
mathematical optimization problem. We first consider the following formulation
of the IMILP as the semi-infinite optimization problem
\begin{equation}
\label{eq:imilp}
\tag{IMILP}
\begin{aligned}
& \min _d  &\quad &\|c-d\|,                 &\quad&    \\
& \st      &      & d^\top x \leq  d^\top x^0 &     &\forall x\in \S.
\end{aligned}
\end{equation}
As in the first formulation, $d$ is a vector of variables, while $c \in \Q^n$
is the given estimated objective function. Note that in~\eqref{eq:imilp}, if
we instead let $x^0$ vary, replacing it with a variable $x$, and interpret $d$
as a fixed objective function, replacing $\|c-d\|$ with the objective $d^\top
x$ of the forward problem~\eqref{eq:milp}, we get a reformulation of the
forward problem~\eqref{eq:milp} itself. This formulation can be made finite,
in the case that $\S$ is bounded, by replacing the possibly infinite set of
inequalities with only those corresponding to the extreme points of
$\conv(\S)$. In the unbounded case, we also need to include inequalities
corresponding to the extreme rays.

Problem~\eqref{eq:imilp} can also be formulated as a conic optimization
problem. Although this is not the traditional way of describing the problem
mathematically, it is arguably the most intuitive representation and is the one
that best highlights the underlying mathematical structure. The following
cones and related sets all play important roles in what follows:
\begin{align*}
& \K := \{ (y,d) \in \Re^{n+1} \mid \|c-d\| \leq y\},\\
& \K(\gamma) := \{ d \in \Re^{n} \mid (\gamma, d) \in \K\},\\
& \K^*(\gamma) := \{x \in \Re ^n \, \mid \, d^\top(x^0-x) \leq 0
\; \forall d \in \K(\gamma) \}, \\
&\C(x^0) := \cone(\S^+ - \{x^0\}) = \cone(\{x-x^0 \mid x \in \S\}),
\textrm{ and}\\
&\D(x^0) := \{d \in \Re^n \mid d^\top (x-x^0) \leq 0 \;\forall x \in \S\}.
\end{align*}
Here, $\K$ is a norm cone, while $\K(\gamma)$ is a ball with center $c$ and
radius $\gamma$ that is a level set of $\K$ and contains vectors whose
objective value in~\eqref{eq:imilp} is at most $\gamma$. $\K^*(\gamma)$
is the set consisting of points, not necessarily in $\S$, that have an objective
function value greater than or equal to that of $x^0$ for \emph{all} the vectors
in $\K(\gamma)$. $\K^*(\gamma)$ can also be seen as the radial cone
obtained by translating the dual of $\cone(\K(\gamma))$ from the origin to
$x^0$ (this is the reason for the slightly abused notation that is typically
used to denote the dual of a cone).

The set $\D(x^0)$ is the set of feasible objectives of~\eqref{eq:imilp} with
target vector $x^0$ and is
precisely the polar of $\C(x^0)$. $\C(x^0)$, on the other hand, is a
translation to the origin of the radial cone that is the intersection of the
half-spaces associated with the facet-defining inequalities valid for
$\conv(\S^+)$ that are binding at $x^0$. Equivalently, it is the radial cone
with vertex at $x^0$ generated by rays $x - x^0$ for all $x \in S$. The
notational dependence on $x^0$ in both sets is for convenience later when
various target vectors will be constructed in the reductions used in the
complexity proofs. We should also point out that $\D(x^0)$ is referred to as
the \emph{normal cone} at $x^0$ in convex analysis and would be denoted as
$N_{\conv(\S^+)}(x^0)$ in the standard notation~\mycitep{rockafellar:70}. Due to
the obvious connections with the theory of polarity in discrete optimization,
however, we maintain our alternative notation here.

Finally, in terms of the cones and sets introduced,~\eqref{eq:imilp} can be
reformulated as 
\begin{equation}
\label{eq:imilp3} \tag{IMILP-C}
\min_{d \in \K(y) \cap \D(x^0)} y.
\end{equation}
\Figurref{fig:intro_ex} illustrates the geometry of the inverse MILP. Here,
$\S$ is the discrete set indicated by the black dots. The estimated objective
function is $c = (0,2)$ and the target vector is $x^0 = (3,1)$. The convex
hull of $\S$ and the cone $\D(x^0)$ (translated to $x^0$) are shaded. The
ellipsoids show the sets of points with a fixed distance to $x^0+c$ for a
given norm. The (unique) optimal solution for this example is vector $d^*$,
which is also the unique point of intersection of $\K(\|c-d^*\|)$ and
$\D(x^0)$. The point $x^0+d^*$ is also illustrated.

\begin{figure}
\centering
\includegraphics{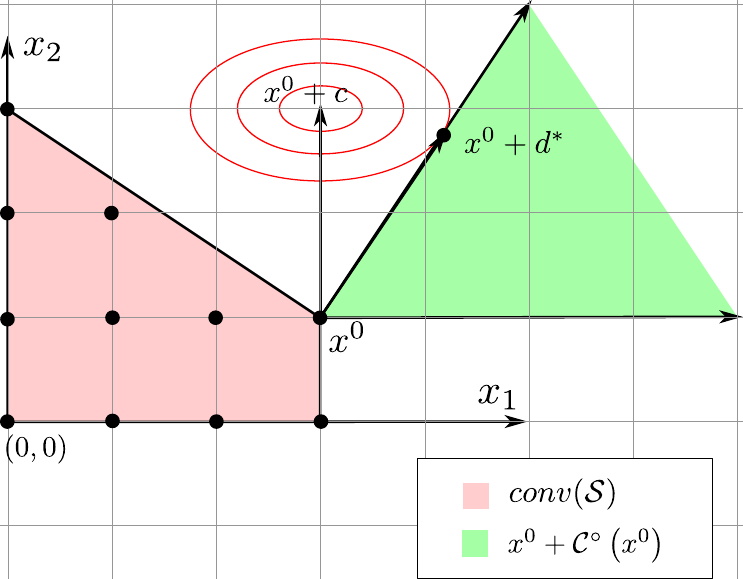}
\caption{Two-dimensional inverse MILP \label{fig:intro_ex}}
\end{figure}

Figure~\ref{fig:sets} presents a more detailed picture of how the various
cones and sets introduced so far are related by displaying sets $\conv(\S)$,
$\D(x^0)$, $\cone(\K(\gamma))$, and $\K^*(\gamma)$ for four different
two-dimensional inverse problems with Euclidean norm.
\begin{figure}[tbh]
\begin{center}
\subfloat[Instance 1\label{fig:sets-1}]{\includegraphics[width=0.43\textwidth]
{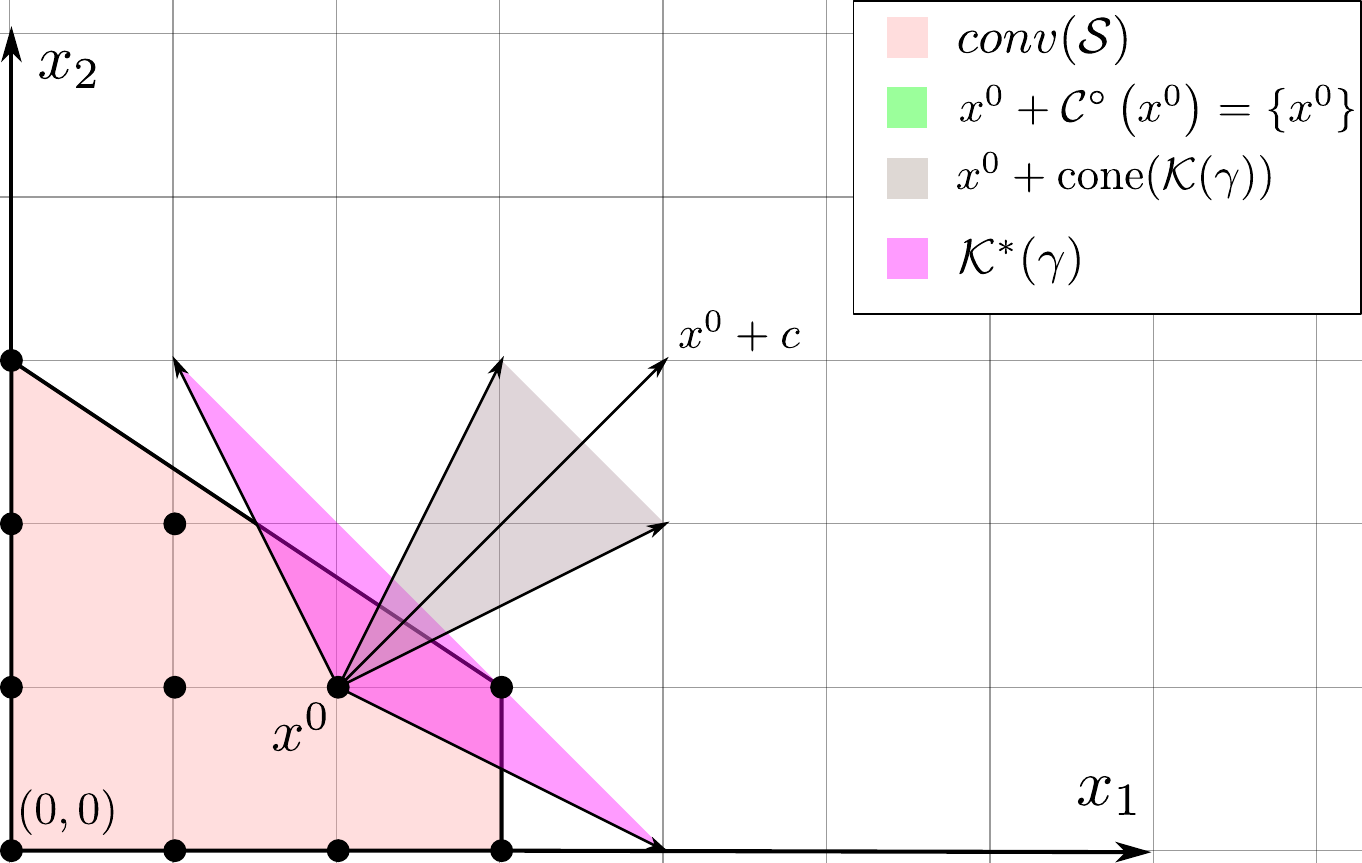}}\;
\subfloat[Instance 2\label{fig:sets-2}]{\includegraphics[width=0.45\textwidth]
{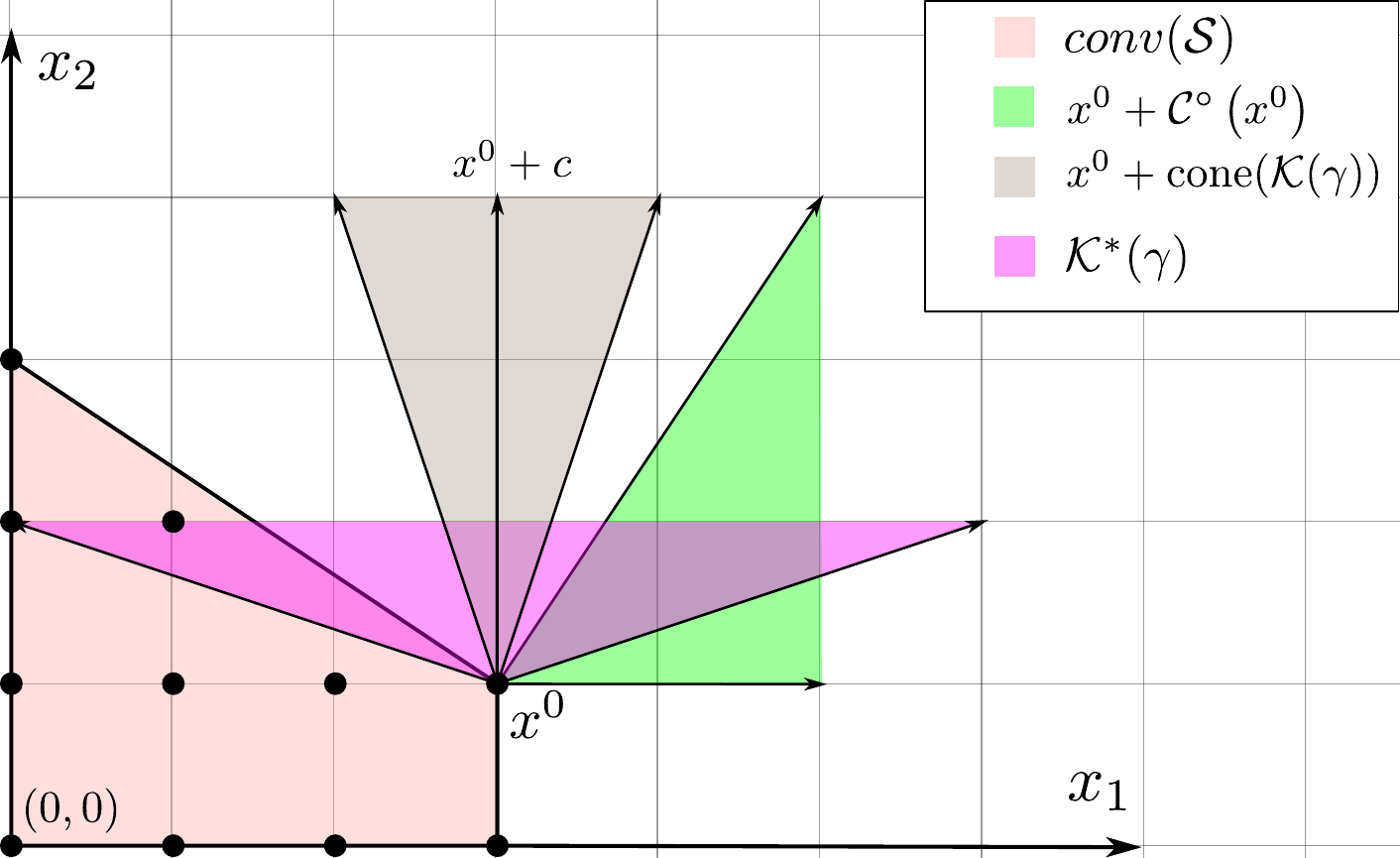}}\\
\subfloat[Instance 3\label{fig:sets-3}]{\includegraphics[width=0.45\textwidth]
{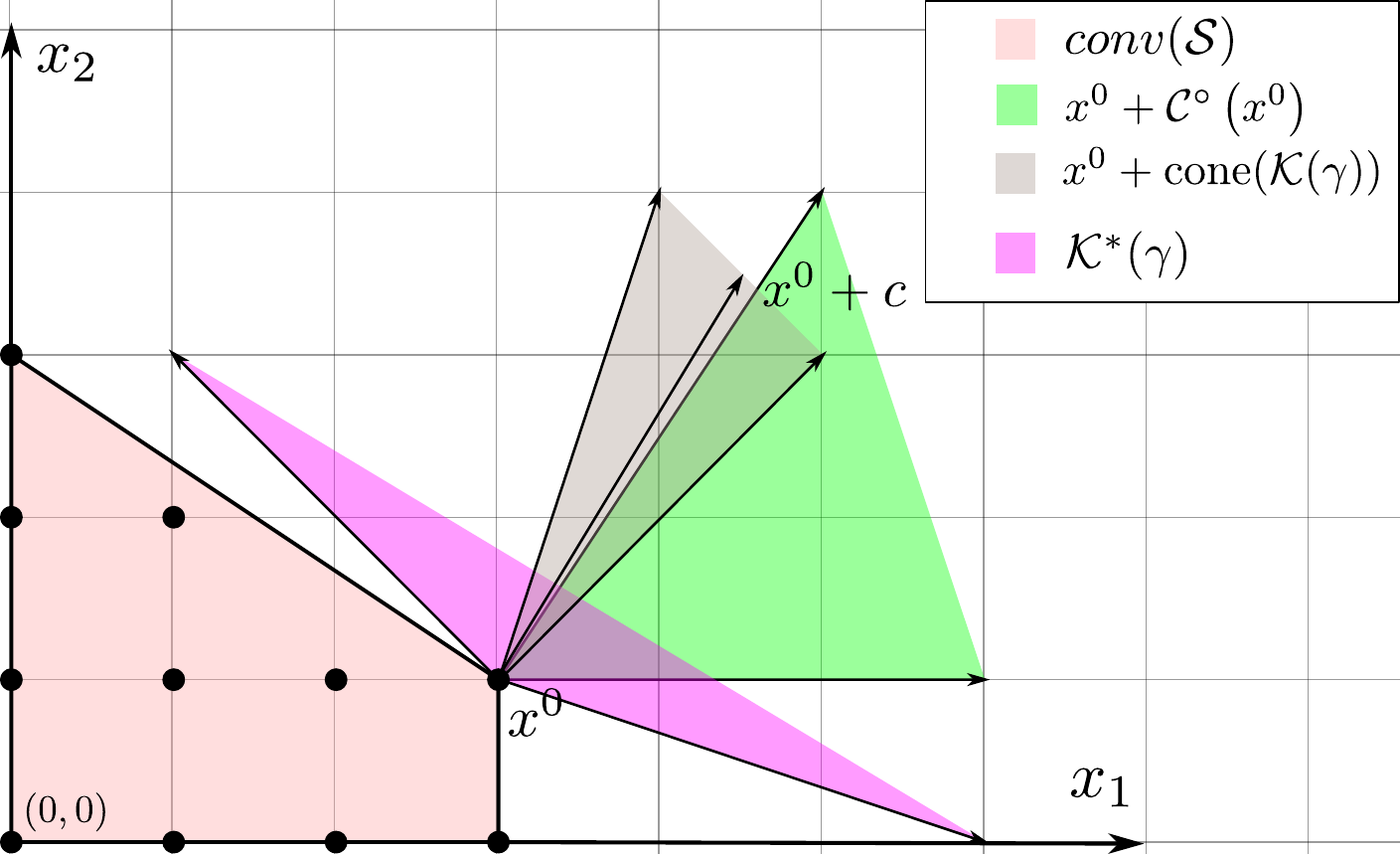}}\;
\subfloat[Instance 4\label{fig:sets-4}]{\includegraphics[width=0.45\textwidth]
{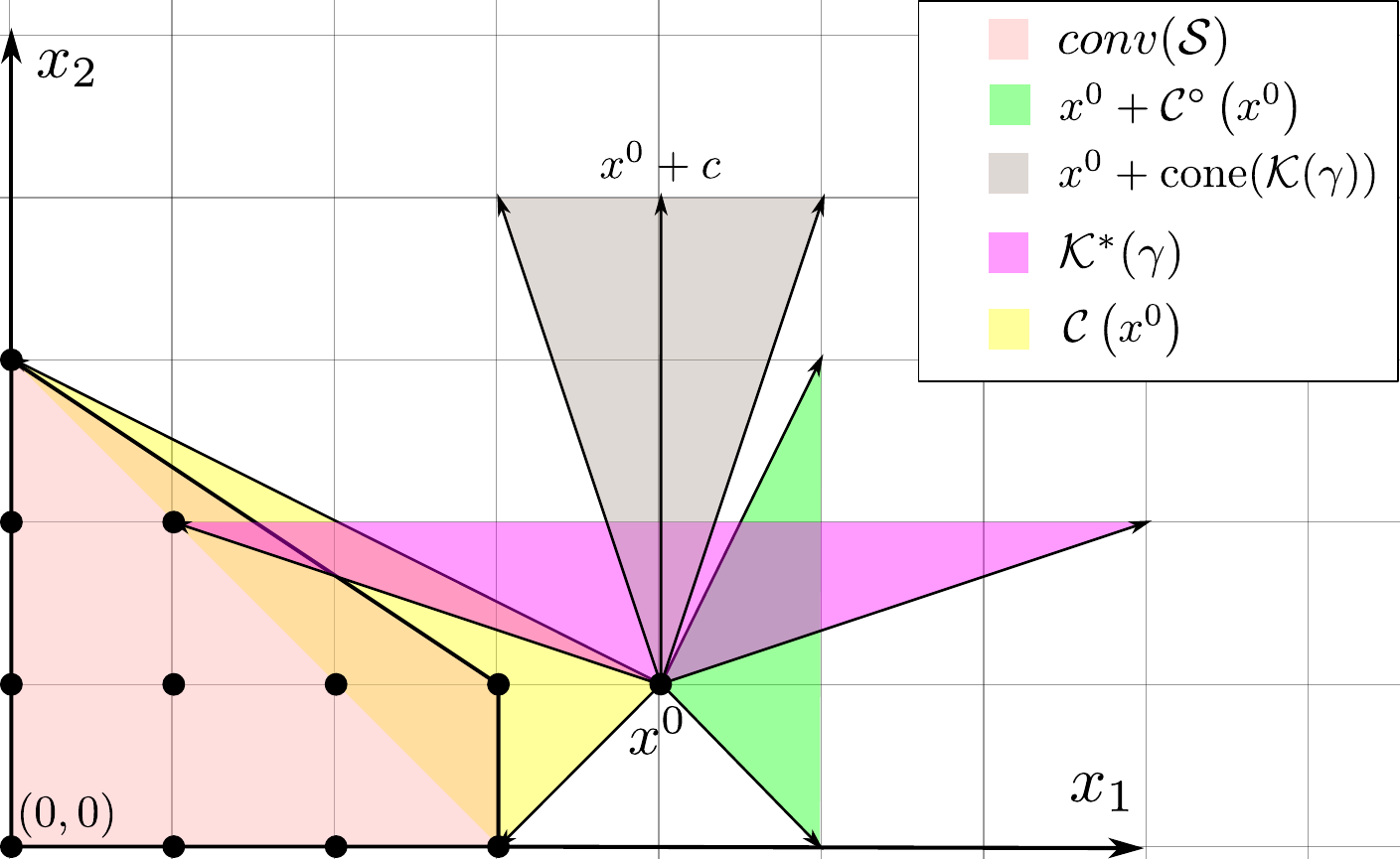}}
\end{center}
\caption{Pictorial illustration of sets $\conv(\S)$, $\D(x^0)$, $\K(\gamma)$,
$\K^*(\gamma)$ for 4 IMILP instances \label{fig:sets}}
\end{figure}
Notions of duality that underlie many of the concepts discussed in the paper
can be seen in the relationships between these sets. $\K^*(\gamma)$ and $\S$
can be thought of as being in the ``primal space'' with respect to the
original problem (the space of primal solutions), whereas the cone $\D(x^0)$
and the ball $\K(\gamma)$ can be thought of as being in the ``dual space,''
the space of directions. In the context of the inverse problem, these roles
are reversed, e.g., $\D(x^0)$ is the set of primal solutions for the inverse
problem.

The interpretation of the feasible region of~\eqref{eq:imilp} as the polar of
$\C(x^0)$ leads to a third formulation in terms of the so-called $1$-polar of
$\conv(\S)$, defined as
\begin{equation*}
\P^1 = \{d \in \Re^n \mid d^\top x \leq 1 \; \forall
x \in \conv(\S)\}.
\end{equation*}
When $\conv(\S)$ is full-dimensional and
$0 \in \inter\left(\conv(\S)\right)$, the 1-polar is the set of all normalized
inequalities valid for $\conv(\S)$ (see~\mycitep{schrijver1986theory} for
definitions). Under these assumptions, \eqref{eq:imilp} can also be
reformulated as
\begin{alignat}{5}
\label{eq:imilp4}
& \min _{d,\rho} &\quad &\|c- \rho d\|      & \quad & \notag \\
& \st   &      & d \in \P ^1,      &       & \notag \\
&       &      & d ^\top x^0 \geq 1, & & \tag{IMILP-1P} \\
&       &      & \rho \geq 0.      &       & \notag
\end{alignat}
In~\eqref{eq:imilp4}, $\rho$ is a multiplier that allows for scaling of the
members of the 1-polar (which are normalized) in order to improve the
objective function value, but otherwise plays no important role. It would seem
to be more natural to require $\|c\| = 1$ or normalize in some other way to
avoid this scaling, but the presence of this scaling variable
highlights that the usual formulation does have this rather unnatural feature.
When $d \in \P^1$, the single constraint $d^\top x^0 \geq 1$ is, in effect,
equivalent to the exponential set of constraints in~\eqref{eq:imilp}, and
ensures $d$ is a feasible objective. Observe also that relaxing the constraint
$d ^\top x^0 \geq 1$ yields a problem similar to the classical separation
problem, but with a different objective function. We revisit this idea
in~\Sectiref{sec:algo}.

Note that when $\conv(\S^+)$ is not full-dimensional, any objective vector in
the subspace orthogonal to the affine space containing $\conv(\S^+)$ is a
feasible objective for the inverse problem. If we let $c_{\S^+}$ be the
projection of $c$ onto the smallest affine space that contains $\S^+$ and
$c_{\S^+} ^{\perp}$ be the projection of $c$ onto the orthogonal subspace, so
that $c = c_{\S^+} + c_{\S^+}^{\perp}$, then whenever $x^0$ is in the relative
interior of $\conv(\S)$, $c_{\S} ^{\perp}$ will be an optimal solution. On the
other hand, when $\conv(\S)$ is full-dimensional, the unique optimal solution
is 0 whenever $x^0$ is in the interior of $\conv(\S)$. 

\subsection{The Separation Problem} The close relationship
between the inverse problem~\eqref{eq:imilp} and the separation problem for
$\conv(\S)$ should already be evident, but we now introduce this idea
formally. Given an $\hat{x} \in \Q^n$, the separation problem for $\conv(\S)$
is to determine whether $\hat{x} \in \conv(\S)$ and, if not, to generate a
hyperplane separating $\hat{x}$ from $\conv(\S)$. When
$\hat{x} \not\in \conv(\S)$ and such a separating hyperplane exists, we can
associate with each such hyperplane a \emph{valid inequality}, defined as
follows.
\begin{definition}
A \emph{valid inequality} for a set $\Qq$ is a pair $(a, b) \in \Q^{n+1}$ such
that $\Qq \subseteq \{x \in \Re^n \mid a^\top x \leq b\}$. The inequality is
said to be \emph{violated by} $\hat{x} \in \Q^n$ if $a^\top \hat{x} > b$.
\end{definition}
Generating a separating hyperplane is equivalent to determining the
existence of $\hat{d} \in \Re^n$ such that
\begin{equation*}
{\hat{d}}^\top \hat{x} > {\hat{d}}^\top x \quad \forall x \in \S.
\end{equation*}
In such a case, $({\hat{d}}^\top, \max_{x \in \S} {\hat{d}}^\top x)$ is
an inequality valid for $\conv(\S)$ that is violated by $\hat{x}$ (therefore
proving that $\hat{x} \not\in \conv(\S)$). On the other hand, $\hat{d}$ is
feasible for~\eqref{eq:imilp} if and only if
\begin{equation*}
{\hat{d}}^\top x \leq {\hat{d}}^\top x^0 \quad  \forall x \in \S,
\end{equation*}
which similarly means that $({\hat{d}}^\top, {\hat{d}}^\top x^0)$ is an
inequality valid for $\conv(\S)$. Thus, feasible solutions to~\eqref{eq:imilp}
can also be viewed as an inequality valid for $\conv(\S)$ binding at $x^0$.
Moreover, when $x^0 \not\in \conv(\S)$, $({\hat{d}}^\top, \max_{x \in \S}
{\hat{d}}^\top x)$ is an inequality valid for $\conv(\S)$ that is violated by
$x^0$. This provides an informal argument for the equivalence
of~\eqref{eq:imilp} and the separation problem. More will be said about this in
the following sections. 

\subsection{Previous Work} There are a range of different flavors of
the inverse optimization problem. The inverse problem we
investigate is to determine objective function coefficients that make a given
solution optimal, but other flavors of inverse optimization include
constructing a missing part of either the coefficient matrix or the right-hand
sides.
\mycite{Hueberger}{Heuberger2004} provides a detailed survey of different
types of inverse combinatorial optimization problems, including types for
which the inverse problem seeks parameters other than objective function
coefficients. A survey of solution procedures for specific combinatorial
problems is provided, as well as a classification of the inverse problems that
are common in the literature. According to this classification, the inverse
problem we study in this paper is an \emph{unconstrained}, \emph{single
feasible objective}, and \emph{unit weight norm} inverse problem.
Our results can be straightforwardly extended to some related cases, such as
multiple given solutions.

\mycite{Cai et al.}{Cai1999} examine an inverse center location problem in
which the aim is to construct part of the coefficient matrix that minimizes the
distances between nodes for a given solution. It is shown that even though the
center location problem is polynomially solvable, this particular inverse
problem is $\NPcomplexity$-hard. This is done by way of a polynomial
transformation of the satisfiability problem to the decision version of the
inverse center location problem. This analysis indicates that the problem of
constructing part of the coefficient matrix is harder than the forward version
of the problem.

\mycite{Huang}{Huang2005} examines the inverse knapsack problem and inverse
integer optimization problems. Pseudopolynomial algorithms for both the
inverse knapsack problem and inverse problems for which the forward problem
has a fixed number of constraints are presented. The latter is achieved by
transforming the inverse problem to a shortest path problem on a directed
graph.

\mycite{Schaefer}{Schaefer2009} studies general inverse integer optimization
problems.
Using super-additive duality, a polyhedral description of the set of all
feasible objective functions is derived. This description has only continuous
variables but an exponential number of constraints. A solution method using
this polyhedral description is proposed.
Finally, \mycite{Wang}{Wang2009114} suggests a cutting-plane algorithm similar
to the one suggested below and presents computational results on several test
problems with an implementation of this algorithm.

The case when the feasible set is an explicitly described polyhedron is
well-studied by \mycite{Ahuja and Orlin}{AhujaSeptember2001}. In their study,
they analyze the shortest path, assignment, minimum cut, and minimum cost flow
problems under the $\ell_1$ and $\ell_{\infty}$ norms in detail. They also
conclude that the inverse optimization problem is polynomially solvable when
the forward problem is polynomially solvable. The present study aims to
generalize the result of Ahuja and Orlin to the case when the forward problem
is not necessarily polynomially solvable, as well as to make connections to
other well-known problems.

In the remainder of the paper, we first introduce a cutting-plane algorithm
for solving~\eqref{eq:imilp} in the case of the $\ell_\infty$ and $\ell_1$
norms. 
In~\Sectiref{sec:algo}, we show that for these norms, the problem can be
expressed as an LP using standard techniques, albeit one with an exponential
number of constraints. The reformulation can then be readily solved using a
standard cutting-plane approach, as observed by~\mycite{Wang}{Wang2009114}. On
the other hand, in~\Sectiref{sec:complexity-IMILP}, we establish the
computational complexity of the problem and show that it is the same for any
$p$-norm. 

\section{A Cutting-plane Algorithm}
\label{sec:algo}

In this section, we describe a basic cutting-plane algorithm for
solving~\eqref{eq:imilp} under the $\ell_1$ and $\ell_{\infty}$ norms. The
algorithm is conceptual in nature and presented in order to illustrate the
relationship of the inverse problem to both the forward problem and the
separation problem. A practical implementation of this algorithm would require
additional sophistication and the development of such an implementation is not
our goal in this paper.

The first step in the algorithm is to formulate~\eqref{eq:imilp} explicitly as
an LP using standard linearization techniques. The objective function of the
inverse MILP under the $\ell_1$ norm can be linearized by the introduction of
variable vector $y$, and associated constraints as shown below.
\begin{mysubequations}{IMILP-L1}
\label{eq:iml1}
\begin{align}
z_{1}^{-1} =
 \min _{d,y,\theta} & \quad \theta                           \nonumber \\
 \st  & \quad \theta = \sum _{i=1} ^n {y_i},  \label{eq:ineq-l1-a} \\
      & \quad c_i - d_i \leq y_i & & \forall i \in \{1,2,\dots,n\},
 \label{eq:ineq-l1-b} \\
      & \quad d_i - c_i \leq y_i & & \forall i \in \{1,2,\dots,n\},
 \label{eq:ineq-l1-c} \\
      & \quad d^\top x \leq d^\top x^0 & & \forall x \in \S.
 \label{eq:opt-ineq-l1}
\end{align}
\end{mysubequations}
For the case of the $\ell_{\infty}$ norm, the variable $\theta$ and two sets
of constraints are introduced to linearize the problem.
\begin{mysubequations}{IMILP-INF}
\label{eq:imli}
\begin{align}
z_{\infty}^{-1} =
\min _{d,\theta} &\quad \theta \nonumber \\
\st  &\quad c_i - d_i \leq \theta      & & \forall i \in \{1,2,\dots,n\},
\label{eq:ineq-inf-a}\\
     &\quad d_i - c_i \leq \theta      & & \forall i \in \{1,2,\dots,n\},
\label{eq:ineq-inf-b}\\
     &\quad d^\top x \leq d^\top x^0 & & \forall x \in \S.
\label{eq:opt-ineq-inf}
\end{align}
\end{mysubequations}
Both~\nativeeqref{eq:iml1} and~\nativeeqref{eq:imli} are continuous,
semi-infinite optimization problems. To obtain a finite problem, one can
replace the inequalities~\eqref{eq:opt-ineq-l1} and~\eqref{eq:opt-ineq-inf}
with constraints~\eqref{eq:ext-points} and~\eqref{eq:ext-rays} involving the
finite set $\E$ of extreme points and $\mathcal{R}$ of rays of the
convex hull of $\S$.
\begin{align}
d^\top x & \leq d^\top x^0 && \forall x \in \E,
\label{eq:ext-points}\\
d^\top r & \leq 0 && \forall r \in \mathcal{R}. \label{eq:ext-rays}
\end{align}
Although constraints~\eqref{eq:ext-points} and~\eqref{eq:ext-rays} yield a
finite formulation, the cardinality of $\E$ and $\mathcal{R}$ may still be
very large and generating them explicitly is likely to be very challenging. It
is thus not practical to construct this mathematical program explicitly via a
priori enumeration. The cutting-plane algorithm avoids explicitly enumerating
the inequalities in the formulation by generating them dynamically in the
standard way.
The algorithm described here is unsophisticated and although versions of it
have appeared in the literature, we describe it again to illustrate the
basic principles at work and to make the connection to a similar existing
algorithm for solving the separation problem.

We describe the algorithm only for the case of~\nativeeqref{eq:imli}, as the
extension to~\nativeeqref{eq:iml1} is straightforward. We assume $\S$ is
bounded, so that $\conv(\S)$ has no extreme rays and $\mathcal{R}=\emptyset$.
As previously observed, \nativeeqref{eq:imli} is an LP with an exponential
class~\eqref{eq:opt-ineq-inf} of inequalities. Nevertheless, the well-known
result of~\mycite{Gr\"otschel et al.}{GroetschelLovaszSchrijver1993} tells us
that~\nativeeqref{eq:imli} can be solved efficiently using a cutting-plane
algorithm, provided we can solve the problem of separating a given point from
the feasible region efficiently. The constraints~\nativeeqref{eq:ineq-inf-a}
and~\nativeeqref{eq:ineq-inf-b} can be explicitly enumerated, so we focus on
generation of constraints~\nativeeqref{eq:opt-ineq-inf}, which means we are
solving the separation problem for set $\D(x^0)$. For an arbitrary
$\hat{d} \in \Re^n$, this separation problem is to either verify that
$\hat{d} \in \D(x^0)$ or determine a hyperplane separating $\hat{d}$ from
$\D(x^0)$.

The question of whether $\hat{d} \in \D(x^0)$ is equivalent to asking whether
$\hat{d}^\top x \leq \hat{d}^\top x^0$ for all $x \in \S$. This can be
answered by determining $x^* \in \argmax_{x \in \S} \hat{d}^\top x$. When
$\hat{d}^\top x^* > \hat{d}^\top x^0$, then $x^*$ yields a new inequality
valid for $\D(x^0)$ that is violated by $\hat{d}$. Otherwise, we have a proof
that $\hat{d} \in \D(x^0)$. Hence, the separation problem for $\D(x^0)$ is
equivalent to the forward problem.

The algorithm alternates between solving a master problem and the separation
problem just described, as usual. The initial master problem is an LP obtained
by relaxing the constraints~\eqref{eq:opt-ineq-inf} in~\nativeeqref{eq:imli}.
After solving the master problem, we attempt to separate its solution from the
set $\D(x^0)$ and either add the violated inequality or terminate, as
appropriate. More formally, the master problem is to determine
\begin{equation}
\label{eq:InvP_k} \tag{\text{Inv$\P_k$}}
\begin{aligned}
\left(d^k, \theta^k\right) \in \argmin_{\left(d,\theta\right)} &\quad \theta \\
\st  &\quad c_i - d_i \leq \theta        & \forall i \in \{1,2,\dots,n\}, \\
     &\quad d_i - c_i \leq \theta        & \forall i \in \{1,2,\dots,n\}, \\
     &\quad d^\top x \leq d^\top x^0 & \forall x \in \E_{k-1} \\
\end{aligned}
\end{equation}
and the separation problem is to determine
\begin{equation}
\label{eq:P_k} \tag{\text{$\P_k$}} x^k \in \argmax _{x \in \S} d^{k\top} x.
\end{equation}
Here, $\E_k = \{x^1, \dots, x^k\}$ are the points in $\S$ generated so far
(which are generally assumed to be extreme points of
$\conv(\S)$, but need not be in general). \eqref{eq:InvP_k} is a relaxation
of~\nativeeqref{eq:imli} 
consisting of only the valid inequalities corresponding to point in $\E_k$.
When~\eqref{eq:P_k} is unbounded, then $x^0$ is in the relative interior of
$\conv(\S)$ and $d=c_{\S} ^{\perp}$ is an optimal solution, as mentioned
earlier. The overall procedure is given in Algorithm~\ref{alg:alg1}.
\begin{algorithm}[tbh]
\caption{Cutting-plane algorithm for~\nativeeqref{eq:imli}}
\begin{algorithmic}[1]
\State $k \leftarrow 0$, $\E^1 \leftarrow \emptyset$.
\Do
  \State $k \leftarrow k+1$.
  \State Solve \eqref{eq:InvP_k} to determine $d^k, \theta^k$.
  \State Solve \eqref{eq:P_k} to determine $x^k$ or show that \eqref{eq:P_k}
  is unbounded.
  \If {\eqref{eq:P_k} unbounded}
    \State $\theta^{*} \leftarrow \left \| c \right\|_{\infty}$, $d^{*} \leftarrow
  0$, STOP. 
  \EndIf
  \State $\E_{k+1} \leftarrow \E_k \cup \{x^k\}$.
\doWhile{$d^{k \top} (x^k - x^0) > 0$}
\State $\theta^{*} \leftarrow \left \| c - d^k \right\|_{\infty}$,
  $d^{*} \leftarrow d^k$, STOP. 
\end{algorithmic}
\label{alg:alg1}
\end{algorithm}

To understand the nature of the algorithm, observe that in iteration $k$, the
master problem is equivalent to an inverse problem in which $\S$ is replaced
$\E_k$. Equivalently, we are replacing $\D(x^0)$ with the restricted set
$\D_k(x^0)$, taken to be the polar of $\C_k(x^0) = \cone(\E_k - \{x^0\})$.
Each member of $\S$ generated corresponds to a member of $\C(x^0)$, so that
the final product of the algorithm is a (partial) description of $\C(x^0)$ and
hence a partial description of $\D(x^0)$. This is analogous to the way in
which a traditional cutting-plane algorithm for solving the original forward
problem generates a partial description of $\conv(\S)$ and highlights the
underlying duality between the H-representation and the V-representation of a
polyhedron (see section~\ref{sec:opt-sep}).

We illustrate by considering a small example. Let
$c=(2,-1)$, $x^0=(0,3)$ and $\S$ be as in Figure~\ref{fig:p}, where both
$x_1$ and $x_2$ are integer and $\conv(\S)$ is shown. The values of $k$, $d^k$,
and $x^k$ in iterations 1--3 are given in Table~\ref{tab:ex}.
\begin{figure}
\centering
\includegraphics{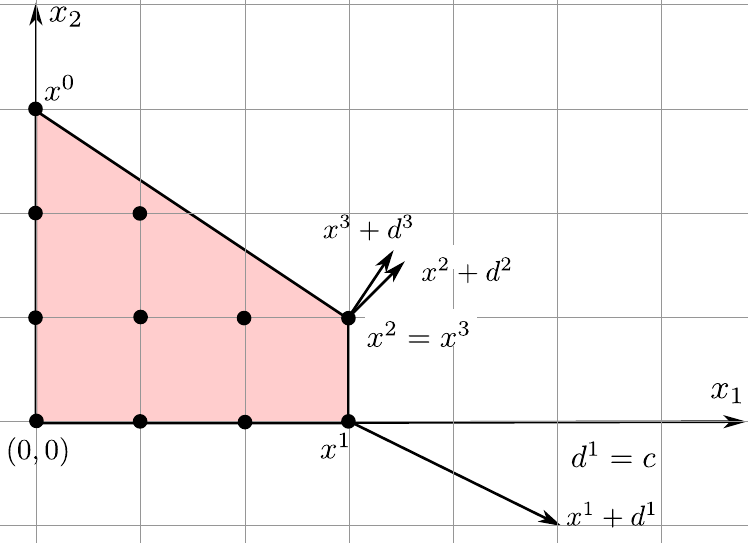}
\caption{Feasible region and iterations of example problem \label{fig:p}}
\end{figure}
\begin{table}[H]
\centering
\begin{tabular}{cccccc}
\hline
                & $k$  & $\E_k$   & $d^k$         & $x^k$     & $\|c-d^k\|_{\infty}$\\
\hline
Initialization  & $1$ & $\emptyset$        & $(2,-1)$      & $(3,0)$  &   $0$ \\
Iteration 1     & $2$ & $\{(3,0)\}$        & $(0.5,0.5)$   & $(3,1)$  & $1.5$\\
Iteration 2     & $3$ & $\{(3,0), (3,1)\}$ & $(0.4,0.6)$   & $(3,1)$  & $1.6$\\\hline
\end{tabular}
\caption{$k$, $d^k$, $x^k$, and $\E_k$ values through iterations \label{tab:ex}}
\end{table}
The (unique) optimal solution of this small example is $d^3=(0.4,0.6)$, and
the optimal value is $\theta^* = \|c-d^3\|_{\infty} = 1.6$.

Figure~\ref{fig:alg1} provides a geometric visualization of another small
example, illustrating how the algorithm would proceed when the set $\S$ is the
collection of integer points inside the polyhedron marked in blue. Here, the
cone $\C_k(x^0)$ is explicitly shown and its expansion can be seen as
generators are added.
\begin{figure}[tbh]
\begin{center}
\subfloat[Iteration 1]{\includegraphics[width=0.33\textwidth]{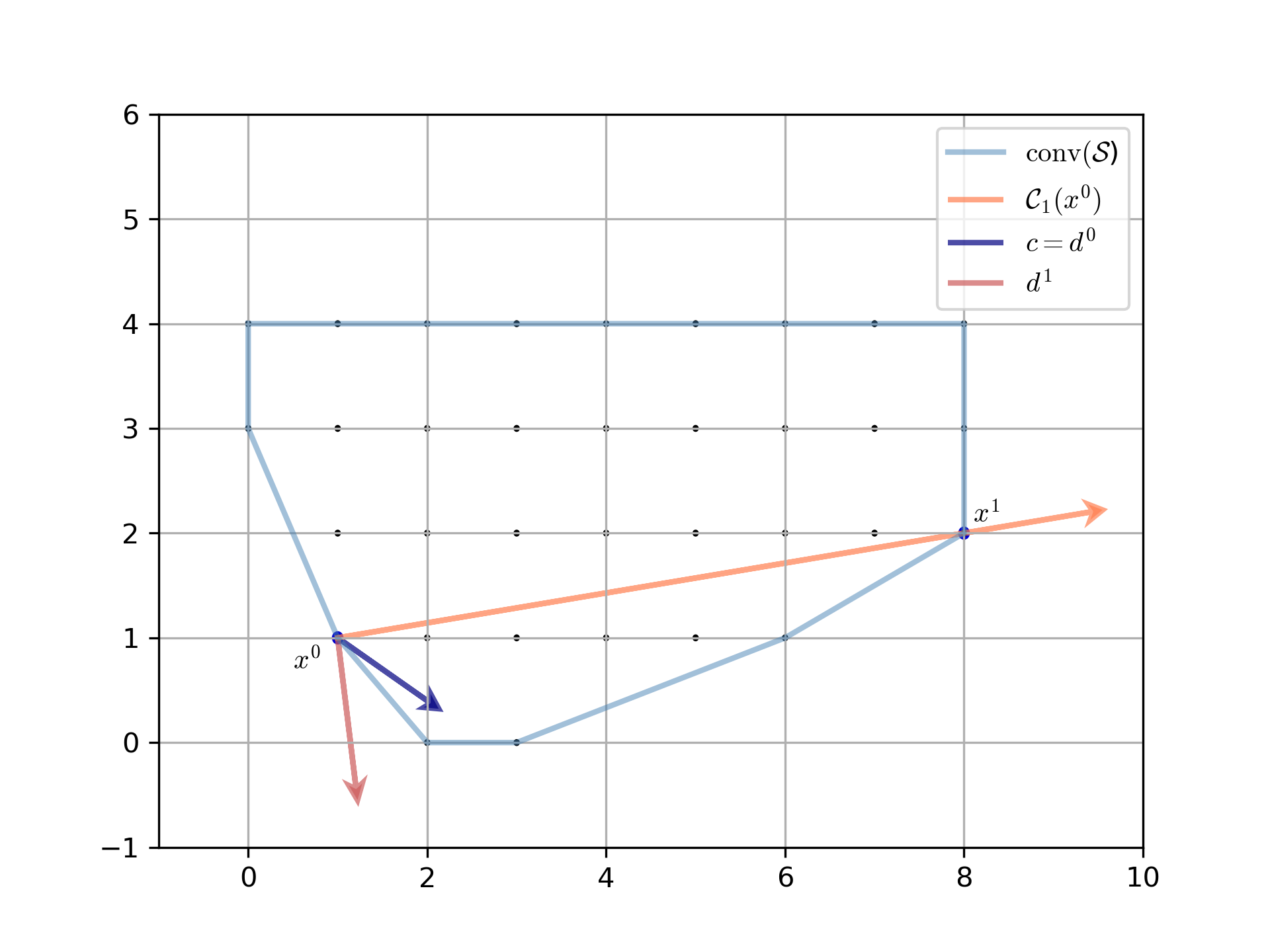}}
\subfloat[Iteration 2]{\includegraphics[width=0.33\textwidth]{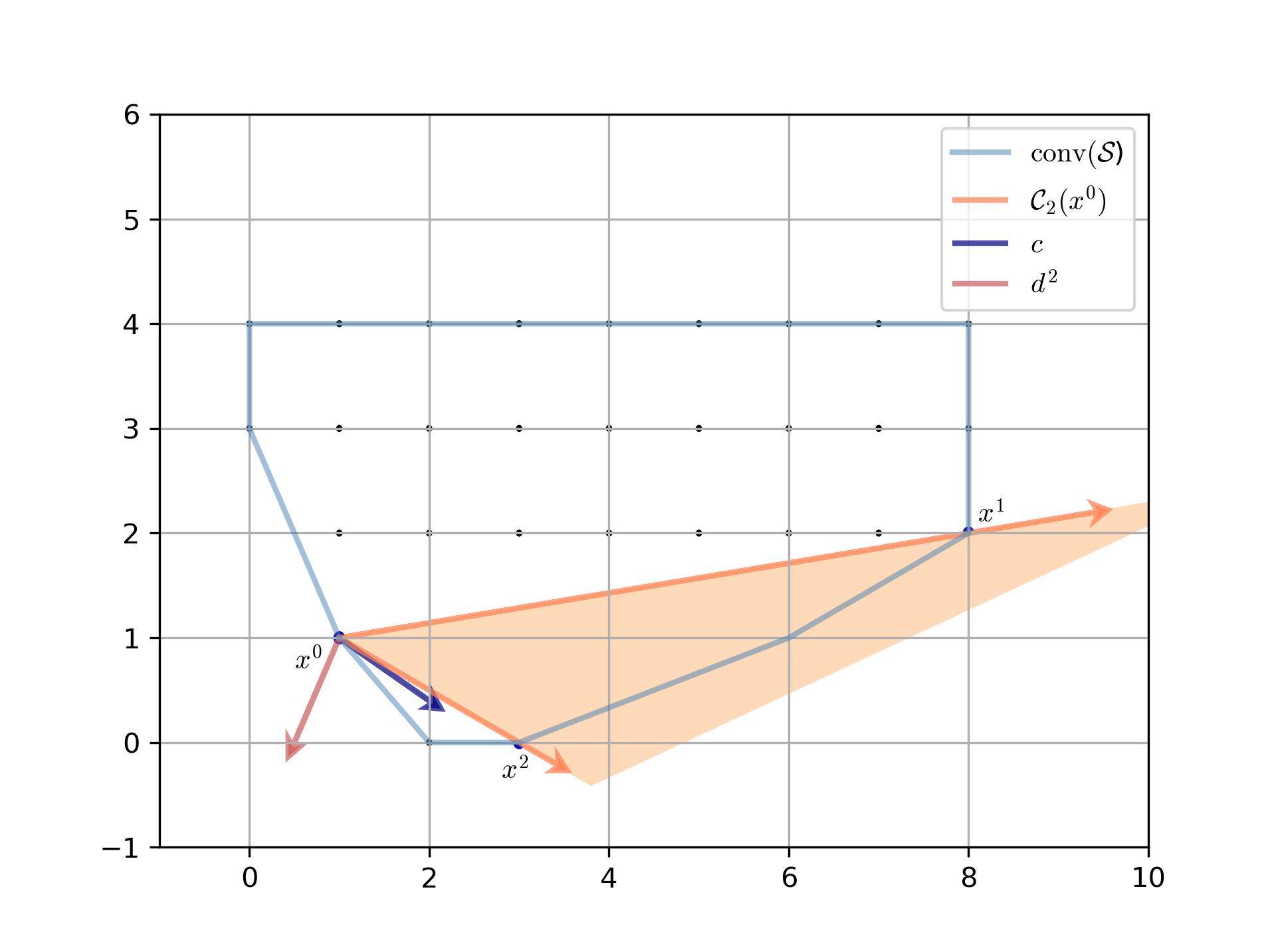}} 
\subfloat[Iteration 3]{\includegraphics[width=0.33\textwidth]{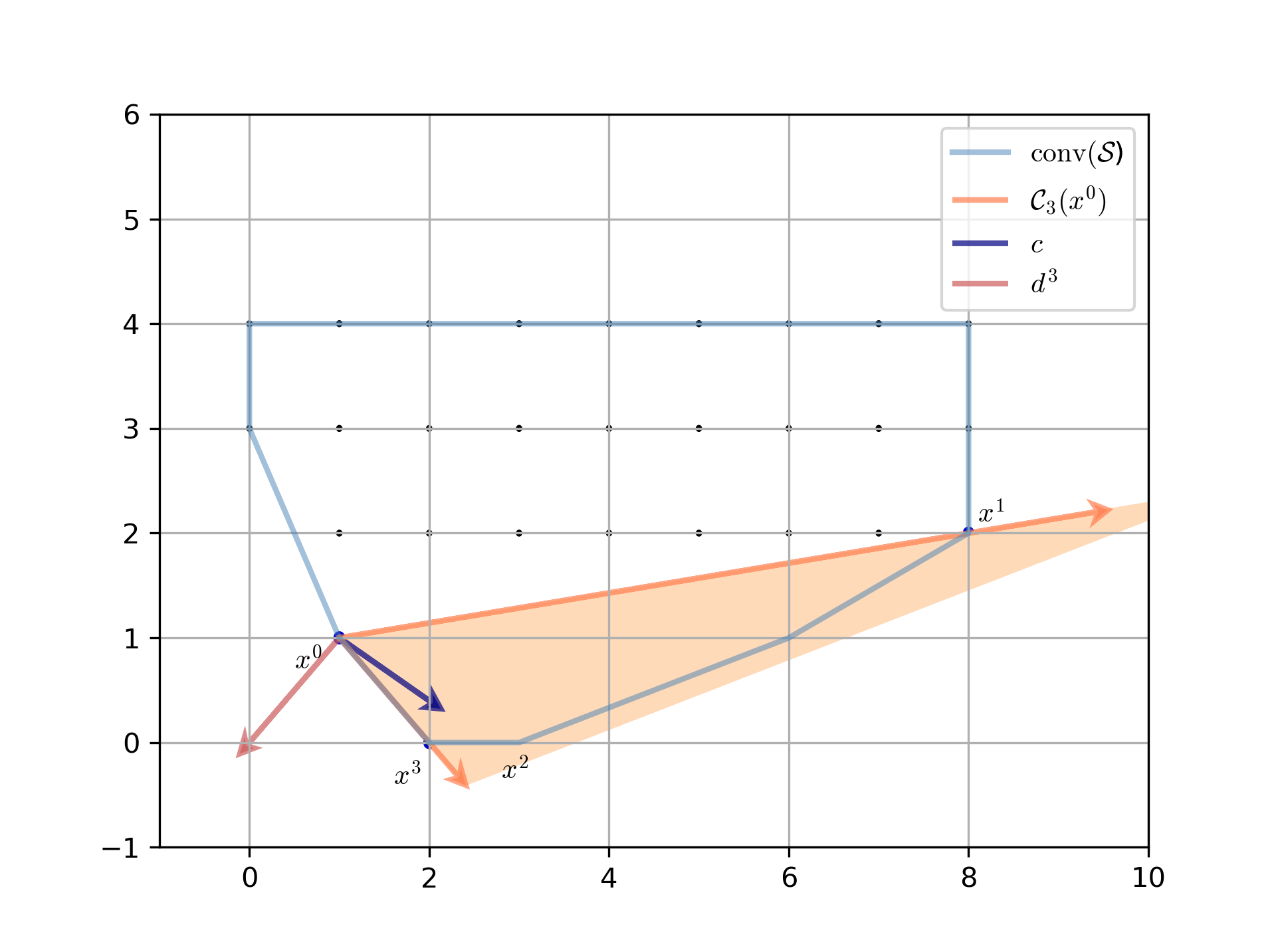}}
\end{center}
\caption{Pictorial illustration of Algorithm~\ref{alg:alg1} \label{fig:alg1}}
\end{figure}

Returning to the relationship of the inverse problem to the separation
problem, observe that an essentially unmodified version of
Algorithm~\ref{alg:alg1} also solves the generic separation problem for
$\conv(\S)$ if we interpret $x^0$ as the point to be separated rather than the
target point. When solving the separation problem,~\eqref{eq:InvP_k} can be
interpreted as the problem of separating $x^0$ from $\conv(\E_k)$. To see
this, note that the dual of~\eqref{eq:InvP_k} is the problem of determining
whether $x^0$ can be expressed as a convex combination of the members of
$\E_k$, i.e., the membership problem for $\conv(\E_k)$. When
$x^0 \not \in \conv(\E_k)$, the Farkas proof of the infeasibility of this LP
is an inequality valid for $\conv(\E_k)$ and violated by $x^0$.~\eqref{eq:P_k}
is then interpreted as the problem of determining whether that same inequality
is valid for the full feasible set $\conv(\S)$, i.e., determining whether
there is a member of $\S$ that violates the inequality, exactly as in the
inverse case.

Figure~\ref{fig:sep} illustrates the application of the algorithm for the
instance from Figure~\ref{fig:alg1}. The only
modification is that we replace the objective function of the master
problem~\eqref{eq:InvP_k} with one measuring the degree of violation of $x^0$,
which is a standard measure of effectiveness for generated valid inequalities.
Even without this modification, a violated valid inequality will be generated,
but the change is to show that the standard separation problem, in which there
is no estimated objective, can also be solved with this algorithm.
Inequalities generated in this way are sometimes called \emph{Fenchel
cuts}~\mycitep{boyd1994}.
\begin{figure}[!h]
\begin{center}
\subfloat[Iteration 1]{\includegraphics[width=0.33\textwidth]{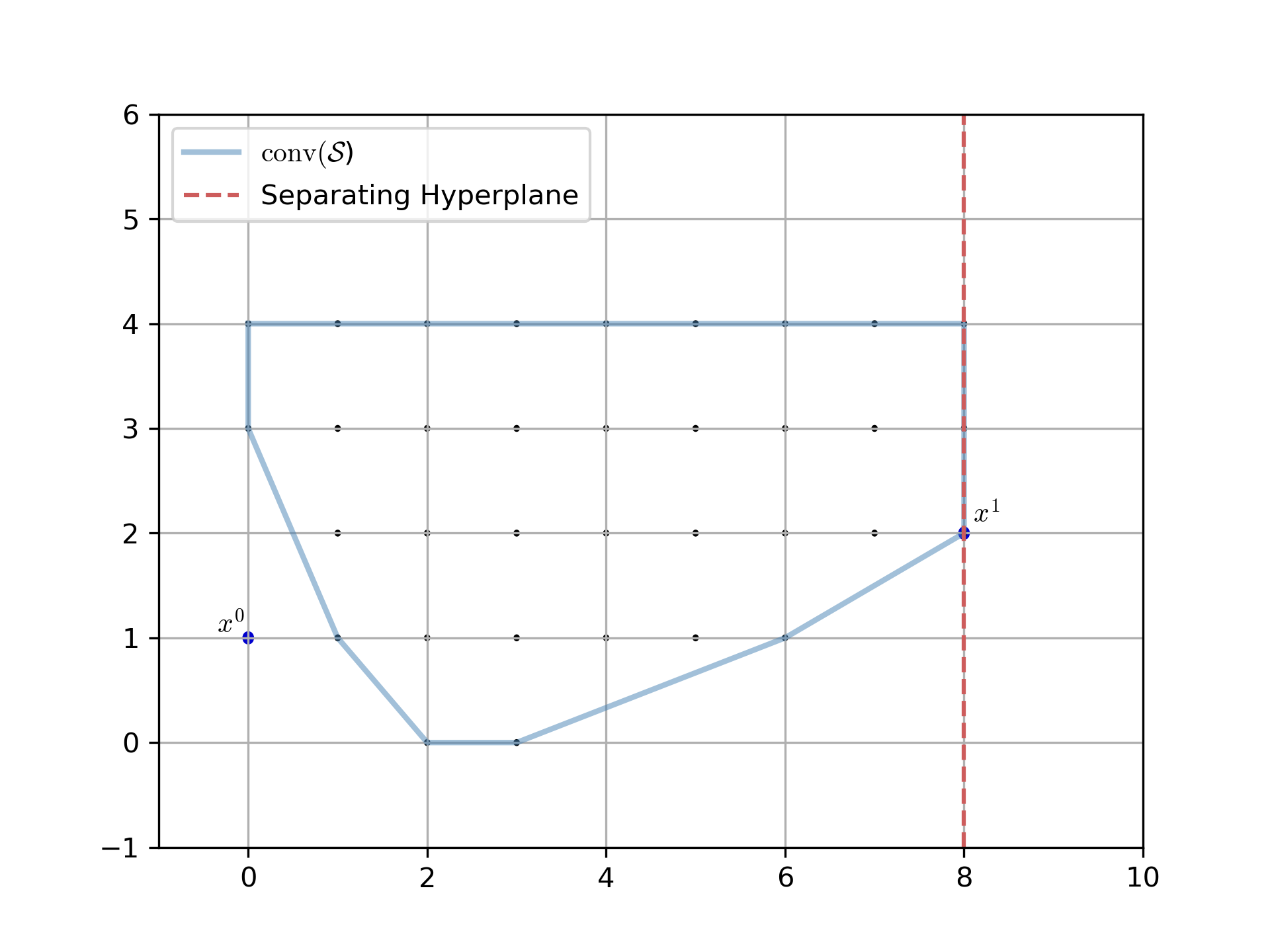}}
\subfloat[Iteration 2]{\includegraphics[width=0.33\textwidth]{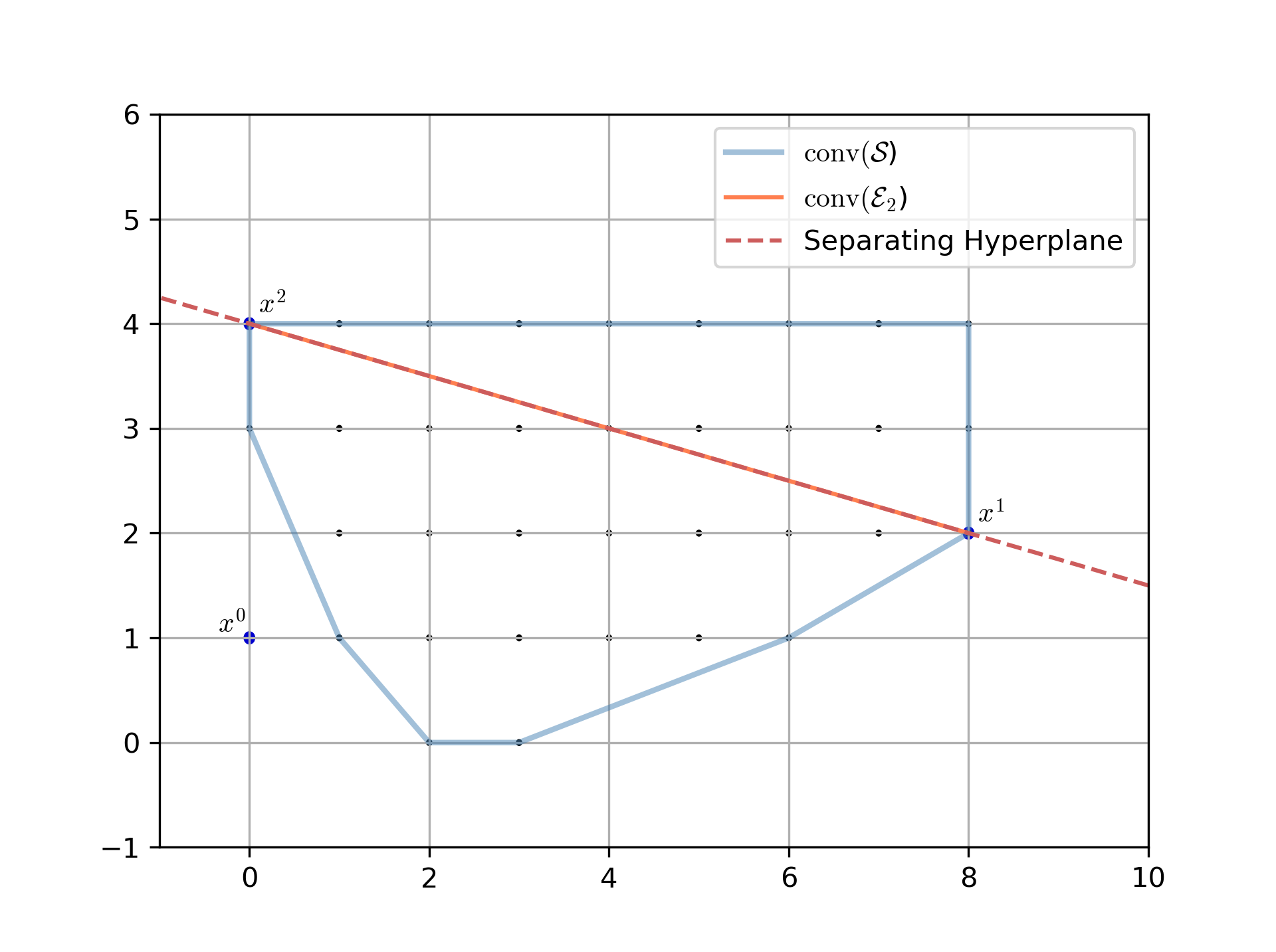}}
\subfloat[Iteration 3]{\includegraphics[width=0.33\textwidth]{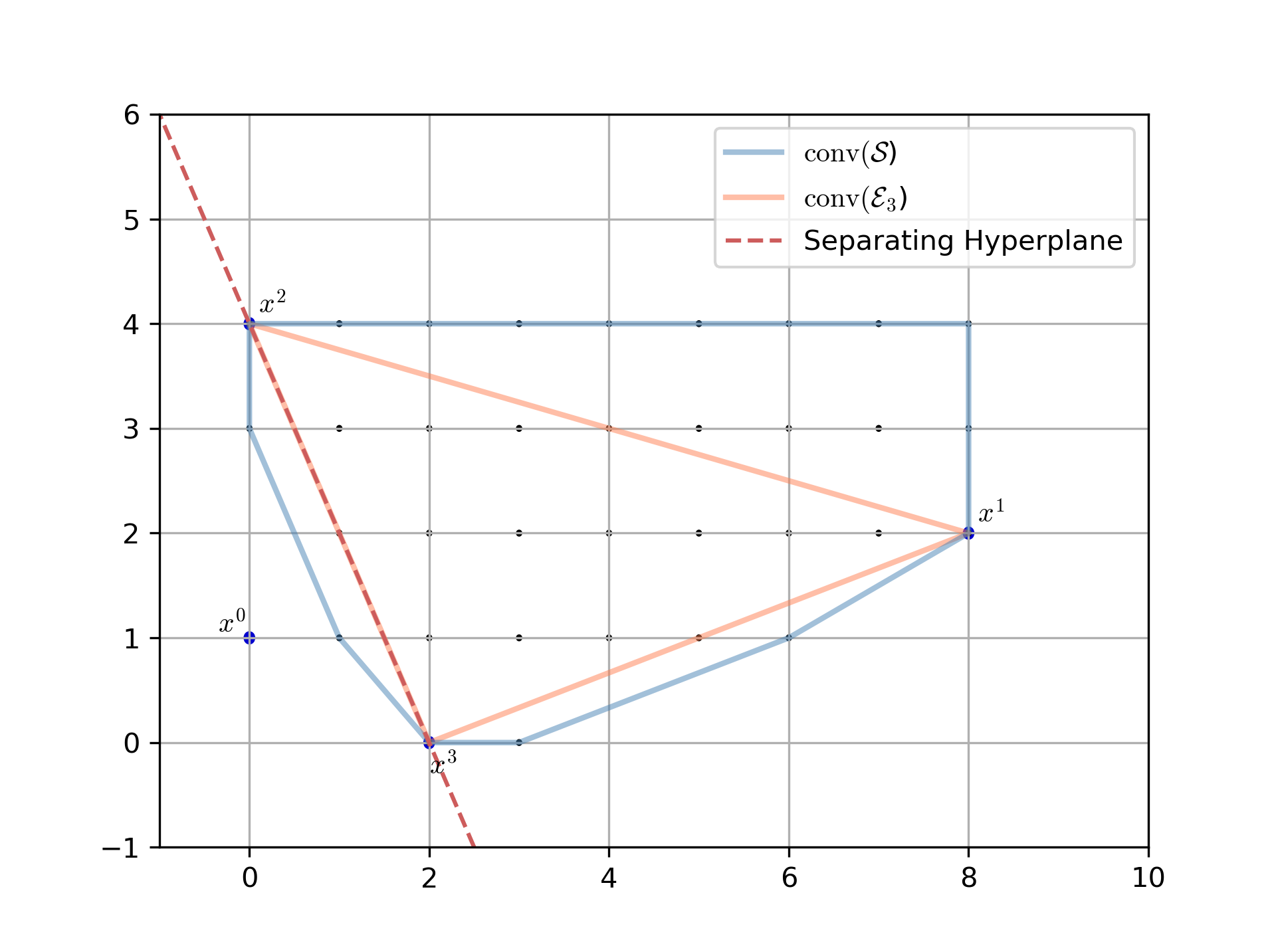}}\\
\subfloat[Iteration 4]{\includegraphics[width=0.33\textwidth]{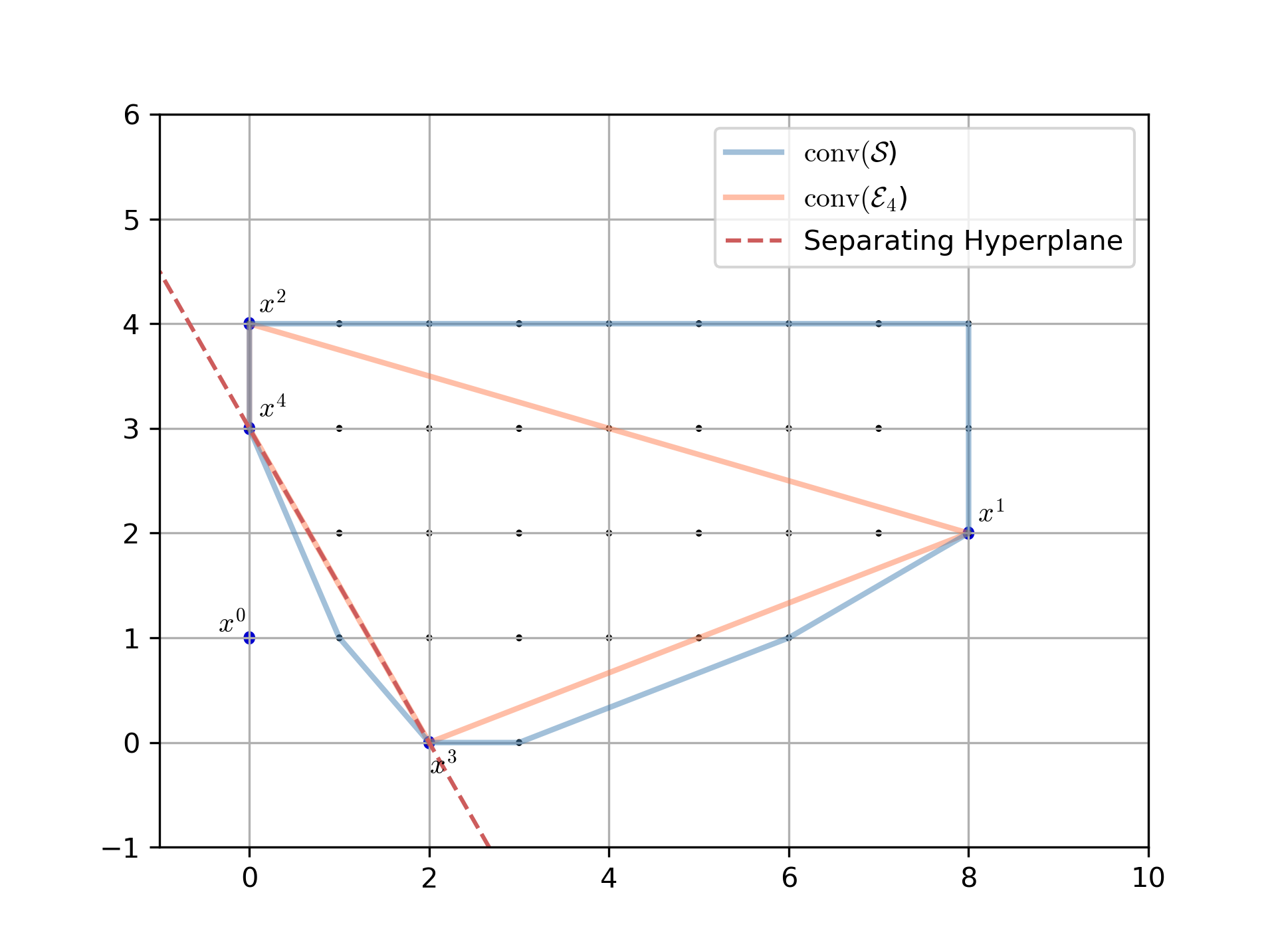}}
\subfloat[Iteration 5]{\includegraphics[width=0.33\textwidth]{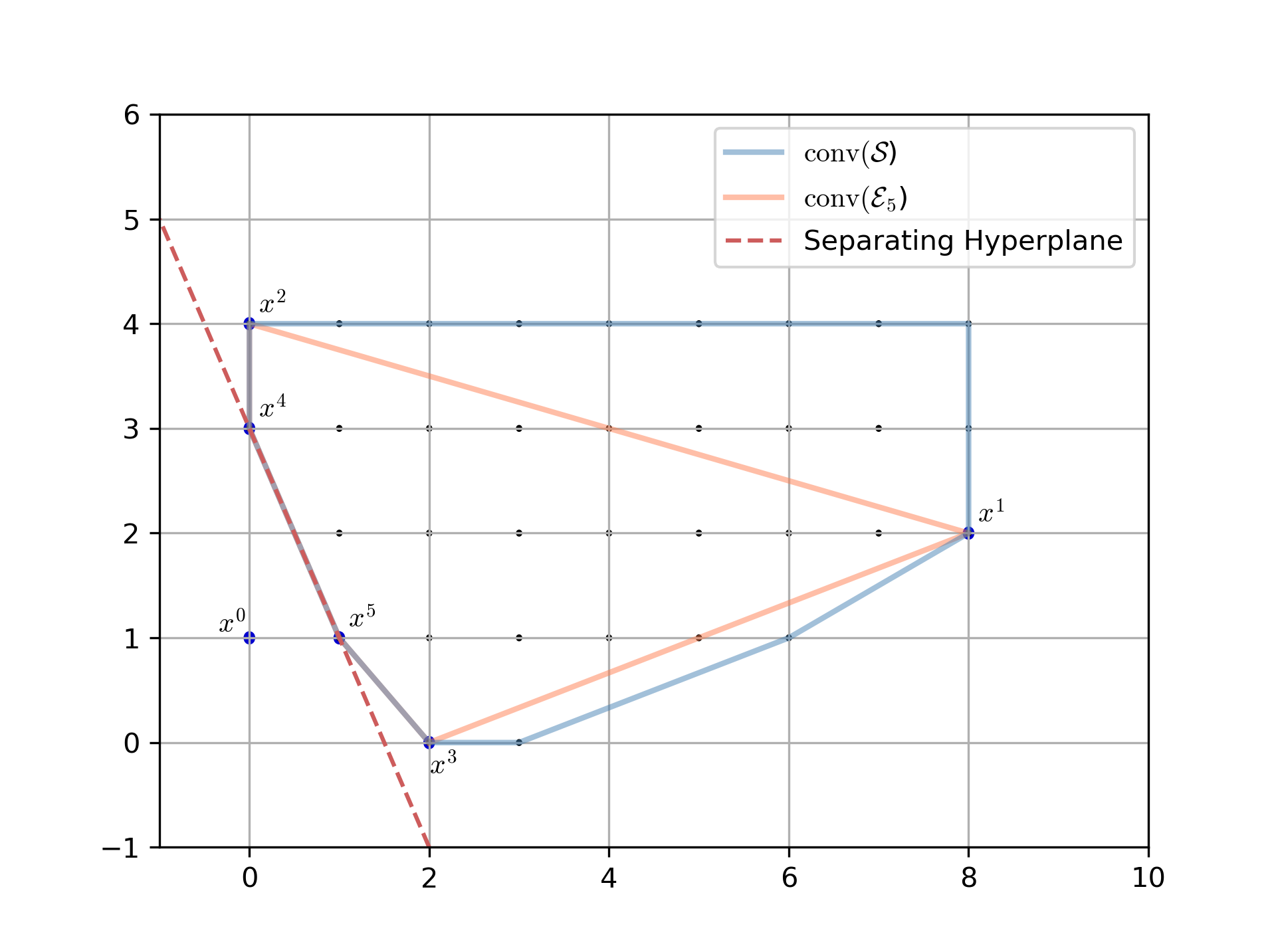}}
\end{center}
\caption{Pictorial illustration of algorithm for generating Fenchel cut
\label{fig:sep}}
\end{figure}

\section{Computational Complexity \label{sec:complexity}}

In this section, we briefly review the major concepts in complexity theory and
the classes into which (the decision versions of) optimizations problems are
generally placed, as well as provide archetypal examples of the kinds of
problems that fall into these classes. We mainly follow the framework
of~\mycite{Garey and Johnson}{Garey:1979:CIG:578533}, but since the material
here will be familiar to most readers, we omit many details and refer the
reader to either~\mycite{Garey and Johnson}{Garey:1979:CIG:578533} or the
sweeping introduction to complexity given by~\mycite{Arora and
Barak}{AroBar07} for a deeper introduction. We provide this brief,
self-contained overview here to emphasize some concepts that are important but
lesser known in the mathematical optimization literature. Among these are the
definitions of the complexity classes $\Dcomplexity^{\Pcomplexity}$ and
$\Delta_2^p$, which play a role in our results below, as well as the
distinction between the polynomial Turing reductions introduced
by~\mycite{Cook}{Cook:1971:CTP:800157.805047} in his seminal work and the
polynomial many-to-one reductions introduced
by~\mycite{Karp}{karp1972reducibility}.

The fundamentals of complexity theory and $\NPcomplexity$-completeness as laid
out in the papers of~\mycite{Cook}{Cook:1971:CTP:800157.805047},
~\mycite{Karp}{karp1972reducibility},~\mycite{Edmonds}{edmonds1971matroids},
and others provide a rigorous framework within which problems arising in
discrete optimization can be analyzed. The origins of the theory can be
traced back to the earlier work on the Entscheidungsproblem
by~\mycite{Turing}{turing1937computable} and perhaps for that reason, it was
originally developed to analyze decision problems, e.g., problems where the
output is YES or NO. Although there exists a theory of complexity that applies
directly to optimization problems~\mycitep{krentel-thesis,KRENTEL1988490,
VOLLMER1995198}), most analyses are done by converting the
optimization problem to an equivalent decision problem form.

The decision problem form typically used for most discrete optimization
problems is that of determining whether a given \emph{primal bound} is valid
(upper bound in the case of minimization or lower bound in the case of
maximization), though we argue later that verification of the optimal value is
more natural. For most problems of current practical interest, verification of
the primal bound is either in the class $\Pcomplexity$ or the class
$\NPcomplexity$. Notable exceptions are the bilevel (and other multilevel)
optimization problems, whose decision versions are in higher levels of the
so-called \emph{polynomial-time hierarchy}\mycitep{Stockmeyer76}.

\subsection{Definitions}

In the framework of~\mycite{Garey and Johnson}{Garey:1979:CIG:578533}, an
algorithm is a procedure implemented using the well-known logic of a
deterministic Turing machine (DTM), a simple model of a computer capable of
sequentially executing a single program (we introduce a ``nondeterministic''
variant below). The input to the algorithm is a string in a
given \emph{alphabet}, which we assume is simply $\{0, 1\}$, since this is the
alphabet on all modern computing devices. As such, the set of all possible
input strings is denoted $\{0, 1\}^*$.
A \emph{problem} (or \emph{problem class}) is defined by describing what set
of input strings (called \emph{instances}) should produce the answer YES. In
other words, each subset $L$ of $\{0, 1\}^*$, formally called
a \emph{language} in complexity theory, defines a different problem for which
algorithms can be developed and analyzed. An algorithm is specified by
describing its implementation as a DTM and is said to \emph{solve} such a
problem if the DTM correctly outputs YES if and only if the input
string is in $L$. In this case, we say the DTM \emph{recognizes} the language
$L$ and that members of $L$ are the instances \emph{accepted} by the DTM.

\myparagraph{Running Time and Complexity} The \emph{running time} of an
algorithm for a given problem is the worst-case number of steps/operations
required by the associated DTM taken over all instances of that
problem. This worst case is usually expressed as a function of the ``size'' of
the input, since the worst case would otherwise be unbounded for any class
with arbitrarily large instances. The size of the input is formally defined to
be its \emph{encoding length}, which is the length of the string representing
the input in the given alphabet. Since we take the alphabet to be $\{0, 1\}$,
the encoding length of an integer $n$ is
\begin{equation*}
\langle n \rangle = 1+ \lceil \log_2(|n| + 1)\rceil.
\end{equation*}
Further, the encoding length of a rational number $r = p/q$ is $\langle
r \rangle = \langle p \rangle + \langle q \rangle$. Encoding lengths play an
important role in the complexity proofs of ~\Sectiref{sec:complexity-IMILP}
below. We discuss these concepts in more detail in Section~\ref{sec:opt-sep},
but refer the reader to the book of~\mycite{Gr\"otschel et
al.}{GroetschelLovaszSchrijver1993} for detailed coverage of definitions and
concepts.

The \emph{computational complexity} of a given problem is the running time
(function) of the ``best'' algorithm, where best is defined by ordering the
running time functions according to their asymptotic growth rate (roughly
speaking, two functions are compared asymptotically by taking the limit of
their ratio as the instance size approaches infinity). For most problems of
practical interest, the exact complexity is not known, so another way of
comparing problems is by placing them into equivalence classes according to
the notion of equivalence yielded by the operation of \emph{reduction}.

\myparagraph{Reduction} A \emph{reduction} is the means by which an algorithm
for one class of problems (specified by, say, language $L_1$) can be used as a
subroutine within an algorithm for another class of problems (specified by,
say, language $L_2$). It is also the means by which problem equivalence and
complexity classes are defined.

In what follows, we refer to two notions of reduction, and the difference
between them is important. The notion that is most relevant in the theory of
NP-completeness is the \emph{polynomial many-to-one reduction} of
\mycite{Karp}{karp1972reducibility}, commonly referred to as \emph{Karp
reduction}. There is a Karp reduction from a problem specified by language
$L_2$ to a problem specified by a language $L_1$ if there exists a mapping
$f:\{0,1\}^* \rightarrow \{0, 1\}^*$ such that
\begin{itemize}
\item $f(x)$ is computable in time polynomial in $\langle x \rangle$ and
\item $x \in L_2$ if and only if $f(x) \in L_1$.
\end{itemize}
Thus, if we have an algorithm (DTM) for recognizing the language
$L_1$ and such a mapping $f$, we implicitly have an algorithm for recognizing
$L_2$. In this case, we say there is a Karp reduction from $L_2$ to $L_1$.

A second notion of reduction is the polynomial Turing reduction, commonly
referred to as \emph{Cook reduction}, introduced
by~\mycite{Cook}{Cook:1971:CTP:800157.805047} in his seminal work. This type
of reduction is defined in terms of \emph{oracles}. An oracle is a conceptual
subroutine that can solve a given problem or class of problems in constant
time. Roughly speaking, the \emph{oracle complexity} of a problem is its
complexity \emph{given the theoretical existence of a certain oracle}. There
is a Cook reduction from a problem specified by language $L_2$ to a problem
specified by language $L_1$ if there is a polynomial-time algorithm for
solving $L_2$ that utilizes an oracle for $L_1$. Hence, the only requirement
is that the number of calls to the oracle must be bounded by a polynomial. The
difference between Karp reduction and Cook reduction is that Karp reduction
can be thought of as allowing only a single call to the oracle as the last
step of the algorithm, whereas Cook reduction allows a polynomial number of
calls to the oracle. There are a range of other notions of reduction that
utilize other different bounds on the number of oracle calls~\mycitep{Kre87}.

Decision problems specified by languages $L_1$ and $L_2$ are said to
be \emph{equivalent} if there is a reduction in both
directions---$L_1$ reduces to $L_2$ and $L_2$ reduces to $L_1$. Equivalence
can be defined using either the Karp or Cook notions of reduction. It is
conjectured (though not known; see~\mycite{Beigel and Fortnow}{BeiFor03}) that
these notions of equivalence are distinct and yield different equivalence
classes of problems. In fact, assuming that $\NPcomplexity \not= \coNP$ (which
is thought to be highly likely), they must be distinct notions, since a
problem specified by any language is trivially seen to be Cook-equivalent to
the problem specified by its complement. To Cook-reduce one problem to the
other, simply solve the complement and negate the answer. Hence, Cook
reduction cannot be used to separate $\NPcomplexity$ from $\coNP$. This
ability to separate $\NPcomplexity$ from $\coNP$ makes Karp reduction a
stronger notion and is part of the rationale for its use as the basis for the
theory of $\NPcomplexity$-completeness in~\mycite{Garey and
Johnson}{Garey:1979:CIG:578533}.

A problem in a complexity class is said to be \emph{complete} for the class if
every other problem in the class can be reduced to it. Informally, this means
that the complete problems are at least as difficult to solve as any other
problem in the class (in a worst-case sense). Completeness of a given problem
in a given complexity class can be shown by providing a reduction from a
problem already known to be a complete problem for the given class. Equivalence,
as described above, is an equivalence relation in the mathematical sense and
can thus be used to define equivalence classes for problems. The complete
problems for a class are exactly those in the largest such equivalence class
that is contained in the class. For the reasons described above, the set of
complete problems is different under Karp and Cook, assuming $\NPcomplexity
\not= \coNP$ (see~\mycite{Lutz and Mayordomo}{LutMay95}). 

\myparagraph{Certificates} Finally, we have the concept of a \emph{certificate}.
A certificate is a string that, when concatenated with the original input
string, forms a (longer) input string to an associated decision problem, which
we informally call the \emph{verification problem}, that yields the same
output as the original one (but can presumably be solved more efficiently). A
certificate can be viewed as a proof of the result of a computation. When
produced by an algorithm for solving the original problem, the certificate
serves to certify the result of that computation after the fact. The
efficiency with which such proofs can be checked is another property of
classes of problems (like the running time) that can be used to partition
problems into classes according to difficulty. We discuss more about the use
of certificates and their formal definition in particular contexts below.

\subsection{Complexity Classes\label{sec:complexity-classes}}

\myparagraph{Class $\Pcomplexity$} The most well-known 
class is $\Pcomplexity$, the class of decision problems that can be solved in
polynomial time on a DTM~\mycitep{Stockmeyer76}. Alternatively, the class
$\Pcomplexity$ can be defined as the \emph{smallest} equivalence class of
problems according to the polynomial equivalence relation described earlier.
Note that for problems in $\Pcomplexity$, there is no distinction between
equivalence according to Karp and Cook. The decision versions of linear
optimization problems (equivalent to the problem of checking whether a system
of inequalities has a solution), the decision versions of minimum cost network
flow problems, and other related problems are all in this class. The
well-known problem of checking whether a system of linear inequalities has a
feasible solution is a prototypical problem in $\Pcomplexity$.

\myparagraph{Class $\NPcomplexity$} $\NPcomplexity$ is the class of problems
that can be solved in polynomial time by a \emph{nondeterministic} Turing
machine (NDTM). Informally, an NDTM is a Turing machine with an infinite
number of parallel processors. With such a machine, if there is a branch in
the algorithm representing two possible execution paths, we can conceptually
follow both branches simultaneously (in parallel), whereas we would need to
explore the branches sequentially in a DTM. A search
algorithm, for example, may be efficiently implemented on an NDTM by following
all possible search paths simultaneously, even if there are exponentially many
of them.

The running time of an algorithm on an NDTM is the number of steps it takes
for \emph{some} execution path to reach an accepting state (a state that
proves the correct output is YES). As a concrete example, consider the problem
of determining whether there exists a binary vector satisfying a system of
linear inequalities. A search algorithm that enumerates the exponentially many
solutions through a simple depth-first recursion would have exponential
running time if implemented using a DTM, while the
running time on an NDTM would be the time to
construct and check the feasibility of one solution.

This last observation leads to an alternative definition of $\NPcomplexity$ as
the class of decision problems for which there exists a certificate for which
there is a DTM that solves the verification problem in time polynomial in the
length of the input when the output is YES. In fact, these two informal
definitions can be formalized and shown to be equivalent.

Intuitively, the idea is that the certificate can be taken to be an encoding
of an execution path that leads to a program state that proves the output is
YES. Thus, the (deterministic) algorithm for verification is similar to the
original nondeterministic algorithm except that it is able to avoid the
``dead ends'' which are explored in parallel in a nondeterministic algorithm.
Formally, if $L \in \NPcomplexity$, then there exists $L^C \in \Pcomplexity$
such that
\begin{align*}
x \in L \Leftrightarrow \;\; & \exists y \in \{0, 1\}^* \text{ such that } (x,
y) \in L^C \text{ and} \\
& \text{$\langle y  \rangle$ is bounded by some polynomial function of
$\langle x \rangle$}.
\end{align*}
In this case, $y$ is the certificate. Because such a certificate has an
encoding length polynomial in the encoding length of $x$ and can be verified
in time polynomial in the encoding length of $x$, such certificates are
sometimes said to be \emph{short} and $\NPcomplexity$ is said to be the class
of decision problems having a ``short certificate.''

In general, problems in $\NPcomplexity$ concern existential questions, such as
whether there exists an element of a set with a given property (alternatively,
whether the set of all elements with a given property is nonempty). Even when
no algorithm for \emph{finding} such an element is known, we may still be able
to efficiently verify that an element given to us has the desired property.
For example, the primal bound verification\footnote{The term ``verification''
is used here in a slightly different way than it is used in the context of
certificates, although the uses are related and the meaning can be generalized
to include both uses.} problem for~\eqref{eq:milp} (usually referred to in the
literature as \emph{the} decision version of MILP) is a prototypical problem
in this class and is defined as follows.
\begin{definition} \label{def:milpd}
  \textbf{MILP Primal Bound Verification Problem (MPVP)}
  \begin{itemize}
  \item INPUT: $\alpha \in \Q$, $d \in \Q^n$,
  $A \in \Q^{m \times n}$, $b \in \Q^m$, and
  $r \in \mathbb{N}$, where $(A, b , r)$ is an encoding of the set
  $\S$ in~\eqref{eq:milp} and $(d, \S)$ is the input
  to~\eqref{eq:milp}.
  \item OUTPUT: YES, if there exists $x \in \S$ such that $d^\top
  x \geq \alpha$, NO otherwise.
  \end{itemize}
\end{definition}
The MPVP is in $\NPcomplexity$ since, 
when the answer is YES, there always exists $x \in \S$ that is itself such a
certificate, i.e., has encoding length polynomially bounded by the encoding
length of the problem input and can be verified in polynomial time to be in
$\S$~\mycitep{MILP-NP}.

The set of problems complete for $\NPcomplexity$ (as defined earlier) is known
simply as $\NPcomplexity$-complete.
The first problem shown to be complete for class $\NPcomplexity$ was the
satisfiability (SAT) problem~\mycitep{Cook:1971:CTP:800157.805047}. It was
proved to be complete by providing a Karp reduction of any problem that
can be solved by an NDTM to the SAT problem.
The MPVP is complete for $\NPcomplexity$ because SAT can be Karp-reduced to it.
It is well-known that the question of whether $\Pcomplexity = \NPcomplexity$
is currently unresolved, though it is widely believed that they are distinct
classes~\mycitep{Aaronson-PvsNP}.

\myparagraph{Class $\coNP$} The class $\coNP$ consists of languages whose
complement is in $\NPcomplexity$. Just as problems in $\NPcomplexity$
typically concern existential questions, problems in $\coNP$ usually concern
the question of whether \emph{all} elements of a set have a given property
(alternatively, whether the set of all elements with a given property is
empty). As such, these problems are not expected to have certificates (for the
YES answer) that can be efficiently verified in the sense defined earlier.
Rather, problems in this class are those for which there exists a string that
can be used to certify the output (in time polynomial in the encoding length
of the input) when the answer is \emph{NO}. For example, we may produce an
element of the set that \emph{doesn't} have the desired property. The MILP
Dual Bound Verification Problem is an example of a prototypical problem in
$\coNP$.
\begin{definition} \textbf{MILP Dual Bound Verification Problem (MDVP)}
  \begin{itemize}
  \item INPUT: $\alpha \in \Q$, $d \in \Q^n$,
  $A \in \Q^{m \times n}$, $b \in \Q^m$, and
  $r \in \mathbb{N}$, where $(A, b , r)$ is an
  encoding of the set $\S$ in~\eqref{eq:milp} and $(d, \S)$
  is the input to~\eqref{eq:milp}.
  \item OUTPUT: YES, if $d^\top x \leq \alpha$ for all $x \in \S$, NO otherwise.
  \end{itemize}
\end{definition}
The input to the MDVP is ($\alpha$, $d$, $A$, $b$, $r$), as in the case of
the MPVP. The MDVP is in $\coNP$ because when the output is NO, there must
be a member of $\S$ with an objective value strictly greater than $\alpha$,
which serves as the certificate. 

\myparagraph{Class $\Dcomplexity^{\Pcomplexity}$} While $\NPcomplexity$ and
$\coNP$ are both well-known classes, the class $\Dcomplexity^{\Pcomplexity}$
introduced by~\mycite{Papadimtriou}{Papadimitriou:1982:CF:800070.802199} is
not as well-known. It is the class of problems associated with languages that
are intersections of a language from $\NPcomplexity$ and a language from
$\coNP$.
A prototypical problem complete for $\Dcomplexity^{\Pcomplexity}$ is the MILP
Optimal Value Verification Problem (MOVP), defined as follows.
\begin{definition}
\label{def:milpv}
\textbf{MILP Optimal Value Verification Problem (MOVP)}
  \begin{itemize}
  \item INPUT: $\alpha \in \Q$, $d \in \Q^n$,
  $A \in \Q^{m \times n}$, $b \in \Q^m$, and
  $r \in \mathbb{N}$, where $(A, b , r)$ is an
  encoding of the set $\S$ in~\eqref{eq:milp} and $(d, \S)$
  is the input to~\eqref{eq:milp}.
  \item OUTPUT: YES, if $\max_{x \in \S} d^\top x := \alpha$, NO otherwise.
  \end{itemize}
\end{definition}
It is easy to see that the language associated with MOVP is the intersection
of the languages of the MPVP and the MDVP, since the output of MOVP is YES
if and only if the outputs of both the MPVP and the MDVP are YES, i.e., $\alpha$
is both a primal and a dual bound for~\eqref{eq:milp}.

In view of our later results, one outcome of this work is the suggestion that
the MOVP is a more natural decision problem to associate with discrete
optimization problems than the more traditional MPVP and should be more widely
adopted in proving complexity results. Most algorithms for discrete
optimization are based on iterative construction of separate certificates for
the validity of the primal and dual bound, which must be equal to certify
optimality. Given certificates for the primal and dual bound verification
problems, a certificate for the optimal value verification problem can thus be
constructed directly. The class $\Dcomplexity^{\Pcomplexity}$ thus contains
the optimal value verification problem associated with~\eqref{eq:milp},
whereas it is the associated MPVP that is contained in the class
$\NPcomplexity$-complete.
We find this is somewhat unsatisfying, since the original
problem~\eqref{eq:milp} is only Cook-reducible to the MPVP.

\myparagraph{The Polynomial Hierarchy} Further classes in the
so-called \emph{polynomial-time hierarchy} ($\ComplexityFont{PH}$), described
in the seminal work of~\mycite{Stockmeyer}{Stockmeyer76}, can be defined
recursively using oracle computation. The notation
$\ComplexityFont{A}^{\ComplexityFont{B}}$ is used to denote the class of
problems that are in $\ComplexityFont{A}$, assuming the existence of an oracle
for problems in class $\ComplexityFont{B}$.

Using this concept, $\Delta_2^p$ is the class of decision problems that can be
solved in polynomial time given an $\NPcomplexity$ oracle, i.e., the class
$\Pcomplexity^\NPcomplexity$. This class is a member of the second level of
$\ComplexityFont{PH}$.
Further levels are defined according to the following recursion.
\begin{align*}
\Delta^p_0 & := \Sigma^p_0 := \Pi^p_0 := \Pcomplexity, \\
\Delta^p_{k+1} & := \Pcomplexity^{\Sigma^p_k}, \\
\Sigma^p_{k+1} & := \NPcomplexity^{\Sigma^p_k}, \text{ and} \\
\Pi^p_{k+1} & := \coNP^{\Sigma^p_k}.
\end{align*}
$\ComplexityFont{PH}$ is the union of all levels of the hierarchy. There is
also an equivalent definition that uses the notion of certificates. Roughly
speaking, each level of the hierarchy consists of problems with certificates
of polynomial size, but whose verification problem is in the class one level
lower in the hierarchy. In other words, the problem of verifying a certificate
for a problem in $\Sigma_i^p$ is a problem in the class
$\Sigma_{i-1}^p$~\mycitep{stockmeyer77}.

\Figurref{fig:classes} illustrates class $\Delta_2^p$
relative to $\Dcomplexity^{\Pcomplexity}$, $\NPcomplexity$, $\coNP$, and
$\Pcomplexity$, assuming $\Pcomplexity \neq \NPcomplexity$. If
$\Pcomplexity=\NPcomplexity$, we conclude that all classes are equivalent,
i.e., $\Delta_2^p = \Dcomplexity^{\Pcomplexity} = \NPcomplexity = \coNP =
\Pcomplexity$. This theoretical possibility is known as the collapse of
$\ComplexityFont{PH}$ to its first
level~\mycitep{Papadimitriou:2003:CC:1074100.1074233} and is thought to be
highly unlikely.
A prototypical problem complete for
$\Delta_2^p$ is the problem of deciding whether a given MILP
has a \emph{unique} solution\mycitep{Papadimitriou:2003:CC:1074100.1074233}.

\begin{figure}
\centering
\includegraphics[scale=0.5]{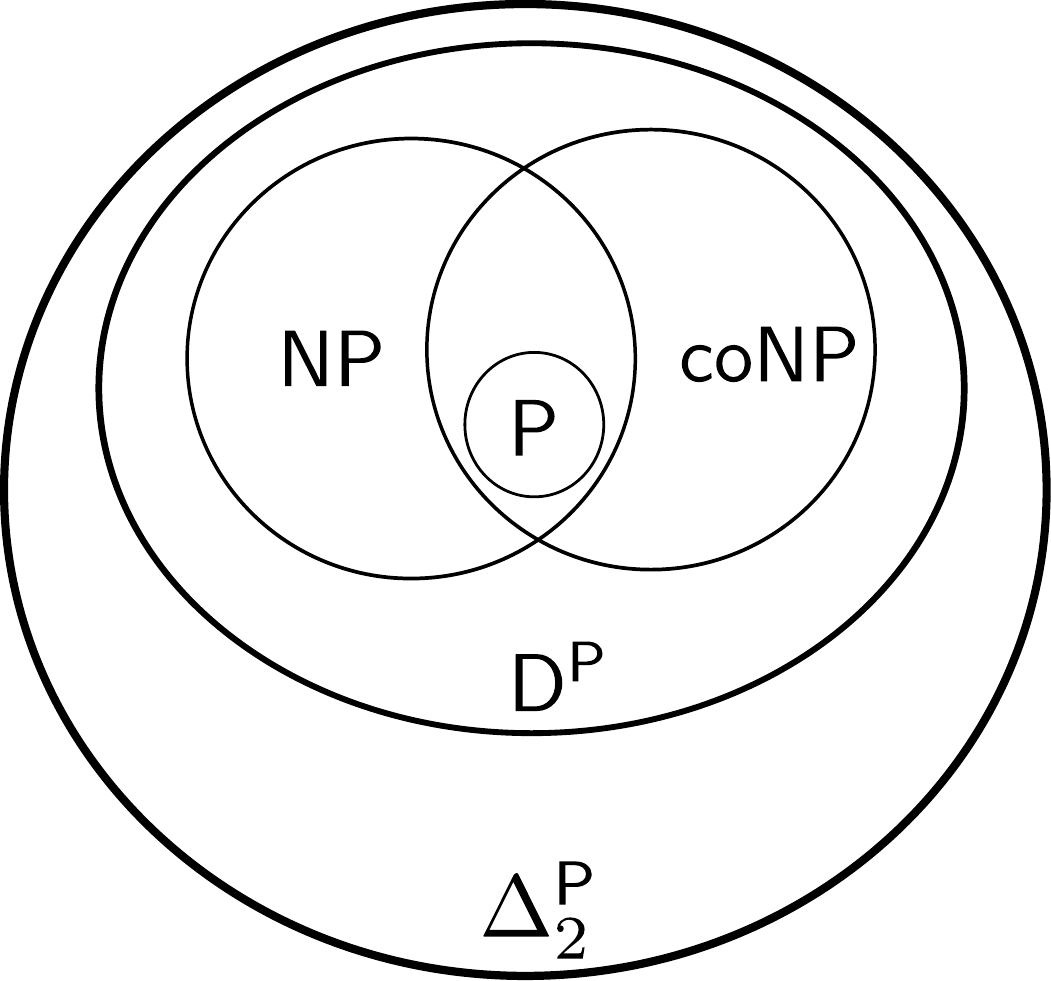}
\caption{Complexity classes $\Delta_2^p$, $\Dcomplexity^{\Pcomplexity}$,
$\NPcomplexity$, $\coNP$, and $\Pcomplexity$, assuming
$\Pcomplexity \neq \NPcomplexity$\label{fig:classes}}
\end{figure}

\subsection{Optimization and Separation\label{sec:opt-sep}}

The concepts of reduction and polynomial equivalence can be extended to
problems other than decision problems, but this requires some additional care
and machinery. Decision problems can be (and often are) reduced to
optimization problems, in a fashion similar to that described earlier, in an
attempt to classify them. When a problem that is complete for a given class in
the polynomial hierarchy can be Karp-reduced to an optimization problem, we
refer to the optimization problem as \emph{hard} for the class. When a
decision problem that is complete for $\NPcomplexity$ can be Karp-reduced to
an optimization problem, for example, the optimization problem is classified
as $\NPcomplexity$-hard. This \emph{does not}, however, require that the
decision version of this optimization problem is a member of the class itself.
In reality, it may be on some higher level of the polynomial hierarchy.
Hardness results can therefore be somewhat misleading in some cases.

In their foundational work, \mycite{Gr\"otschel, Lov{\'a}sz and
Schrijver}{GroetschelLovaszSchrijver1993}
develop a detailed theoretical basis
for the claim that the separation problem for an implicitly defined polyhedron
is polynomially equivalent to the optimization problem over that same
polyhedron.
This is done very carefully, beginning from certain decision problems and
proceeding to show their equivalence to related optimization problems. The
notion of reduction used, however, is Cook reduction, which means that the
results do not directly tell us whether there are decision versions of the
problems discussed that are complete for the same class within
$\ComplexityFont{PH}$. The results discussed below, in contrast, use Karp
reduction to show that there are decision versions of the inverse and forward
optimization problems that are complete for the same complexity class within
$\ComplexityFont{PH}$.

Using an optimization oracle to solve the separation problem (and vice versa)
involves converting between different representations of a given polyhedron.
In particular, we consider (implicit) descriptions in terms of the half-spaces
associated with facet-defining valid inequalities, so-called
\emph{H-representations}, and in terms of generators (vertices and
extreme rays), so-called \emph{V-representations}. There is a formal
mathematical duality relating these two forms of polyhedral representation,
which underlies the notion of polarity mentioned earlier and is at the heart
of the equivalence between optimization and separation. It is this very same
duality that is also at the heart of the equivalence between optimization
problems and their inverse versions. This can be most clearly seen by the fact
that the feasible region $\D(x^0)$ of the inverse problem is the polar of
$\C(x^0)$, a conic set that contains $\conv(\S)$, as we have already
described.

The framework laid out by~\mycite{Gr\"otschel
et al.}{GroetschelLovaszSchrijver1993} emphasizes 
that the efficiency with which the various representations can be manipulated
algorithmically depends inherently and crucially on, among other things, the
encoding length of the elements of these representations. For the purposes of
their analysis,
\mycite{Gr\"otschel et al.}{GroetschelLovaszSchrijver1993} defined the notions of
the \emph{vertex complexity} and \emph{facet complexity} of a polyhedron,
which we repeat here, due to their relevance in the remainder of the paper.
\begin{definition}{(\mycite{Gr\"otschel et al.}{GroetschelLovaszSchrijver1993})}
\begin{enumerate}[label=\emph{(}\roman*\emph{)}]
\item A polyhedron $\P \subseteq \Re^n$
has \emph{facet complexity} of at most $\varphi$ if there exists a rational
system of inequalities (known as an H-representation) describing the
polyhedron in which the encoding length of each inequality is at most
$\varphi$.
\item Similarly, the \emph{vertex complexity} of $\P$ is at most $\nu$ if
there exist finite sets $V, E$ such that $\P = \conv(V) + \cone(E)$ (known as a
V-representation) and such that each of the vectors in $V$ and $E$ has
encoding length at most $\nu$.
\end{enumerate}
\end{definition}
It is important to point out that these definitions are not given in terms of
the encoding length of a full description of $\P$ because $\P$ may be an
implicitly defined polyhedron whose description is never fully constructed.
What \emph{is} explicitly constructed are the components of the description
(extreme points and facet-defining inequalities). The importance of the facet
complexity and vertex complexity in the analysis is primarily that they
provide bounds on the norms of these vectors. The ability to derive such
bounds is a crucial element in the overall framework they present. The
following are relevant results from~\mycite{Gr\"otschel
et al.}{GroetschelLovaszSchrijver1993}.
\begin{proposition}{(\mycite{Gr\"otschel
et al.}{GroetschelLovaszSchrijver1993})} \label{prop:encoding} 
\begin{enumerate}[label=\emph{(}\roman*\emph{)}]
\item (1.3.3) For any $r \in \Q$, $2^{-\langle r \rangle + 2} \leq |r| \leq
2^{\langle r \rangle -1} - 1$. 
\item (1.3.3) For any $x \in \Q^n$, $\|x\|_p < 2^{\langle x \rangle -
n}$ for $p \geq 1$. 
\item (6.2.9) If $\P$ is a polyhedron with vertex complexity at most $\nu$
and $(a, b) \in \Z^{n+1}$ is such that
\begin{equation*}
a^\top x \leq b + 2^{-\nu-1}
\end{equation*}
holds for all $x \in \P$, then $(a, b)$ is valid for $\P$.
\end{enumerate}
\end{proposition}
In other words, the facet complexity and the encoding lengths of the vectors
involved specify a ``granularity'' that can allow us to, for example, replace
a ``$<$'' with a ``$\leq$'' if we can bound the encoding length of the
numbers involved.

Utilizing the above definitions, the result showing the equivalence of
optimization and separation can be formally stated as follows.
\begin{theorem}{(\mycite{Gr\"otschel et al.}{GroetschelLovaszSchrijver1993})}
Let $\P \subseteq \Re^n$ be a
polyhedron with facet complexity $\varphi$. Given an oracle for any one of
\begin{itemize}
\item the dual bound verification problem over $\P$ with linear objective
coefficient $d \in \Q^n$, 
\item the separation problem for $\P$ with respect to $\hat{x} \in \Q^n$, or
\item the primal bound verification problem for $\P$ with linear objective
coefficient $d \in \Q^n$, 
\end{itemize}
there exists an oracle polynomial-time algorithm for solving either of the
other two problems. Further, all three problems are solvable in time
polynomial in $n$, $\varphi$, and either $\langle d \rangle$ (in the case of
primal or dual bound verification) or $\langle \hat{x} \rangle$ (in the case of
separation).
\label{th:gls}
\end{theorem}
The problem of verifying a given dual bound was called the \emph{violation
problem} by~\mycite{Gr\"otschel et al.}{GroetschelLovaszSchrijver1993}. The
above result refers only to the facet complexity $\varphi$, but we could also
replace it with the vertex complexity $\nu$, since it is easy to show that
$\nu \leq 4n^2\varphi$. In the remainder of the paper, we refer to the
``polyhedral complexity'' whenever the facet complexity and vertex complexity
can be used interchangeably.

\section{Complexity of Inverse MILP~\label{sec:complexity-IMILP}}

In this section, we apply the framework discussed in~\Sectiref{sec:complexity}
to analyze the complexity of the inverse MILP. We follow the traditional
approach and describe the complexity of the decision versions. In addition to
the standard primal bound verification problem, we also consider the dual
bound and optimal value verification problems. We show that the primal bound,
dual bound, and optimal value verification problems for the inverse MILP are
in the complexity classes $\coNP$-complete, $\NPcomplexity$-complete, and
$\Dcomplexity^{\Pcomplexity}$-complete, respectively. Hence, the classes
associated with primal and dual bounding problems are reversed when going from
MILP to inverse MILP, while the optimal value verification problems are both
contained in the same class. 

\subsection{Polynomially Solvable Cases}

The result of~\mycite{Ahuja and Orlin}{AhujaSeptember2001} can be applied
directly to observe that there are cases of the inverse MILP that are
polynomially solvable. In particular, they showed that the inverse problem can
be solved in polynomial time whenever the forward problem is polynomially
solvable.
\begin{theorem}{(\mycite{Ahuja and Orlin}{AhujaSeptember2001})} \label{th:ahuja}
If an optimization problem is polynomially solvable for each linear cost
function, then the corresponding inverse problems under the $\ell_1$ and
$\ell_{\infty}$ norms are polynomially solvable.
\end{theorem}
\mycite{Ahuja and Orlin}{AhujaSeptember2001} use Theorem~\ref{th:gls}
from~\mycite{Gr\"otschel et al.}{GroetschelLovaszSchrijver1993} to conclude
that inverse LP, in particular, is polynomially solvable. The separation
problem in this case is an LP of polynomial input size and is hence
polynomially solvable. The theorem of ~\mycite{Gr\"otschel
et al.}{GroetschelLovaszSchrijver1993} shows that both problems are in the
class $\Pcomplexity$, since Karp and Cook reductions are equivalent for
problems in $\Pcomplexity$. \Theorref{th:ahuja} also indicates that if a given
MILP is polynomially solvable, then the associated inverse problem is also
polynomially solvable.

\subsection{General Case}

In the general case, the MILP constituting the forward problem is not known to
be polynomially solvable and so we now consider MILPs whose decision versions
are complete for $\NPcomplexity$.
Applying the results of~\mycite{Gr\"otschel
et al.}{GroetschelLovaszSchrijver1993} straightforwardly, as ~\mycite{Ahuja and
Orlin}{AhujaSeptember2001} did, we can easily show that~\nativeeqref{eq:iml1}
and~\nativeeqref{eq:imli} can be solved in polynomial time, given an oracle
for the MPVP, as stated in the following theorem.
\begin{theorem}
Given an oracle for the MPVP,~\nativeeqref{eq:iml1} and~\nativeeqref{eq:imli} are
solvable in time polynomial in $n$, the vertex complexity of $\conv(\S^+)$,
and $\langle c \rangle$.
\label{th:gls-informal}
\end{theorem}

The above result directly implies that IMILP under $\ell_1$ and $\ell_\infty$
norms is in fact in the complexity $\Delta_2^{\Pcomplexity}$, but stronger
results are possible, as we show. In the remainder of this section, we assume
the norm used is a $p$-norm, as this is needed for some results (in
particular, Proposition~\ref{prop:encoding} is crucially used).

\myparagraph{Definitions} We next define decision versions of the inverse MILP
analogous to those we defined in the case of MILP. These similarly attempt to
verify that a given bound on the objective value is a primal bound, a dual
bound, or an exact optimal value.
The primal bound verification problem for inverse MILP is as follows.
\begin{definition}
\label{def:invd}
\textbf{Inverse MILP Primal Bound Verification Problem
(IMPVP)}:
  \begin{itemize}
  \item INPUT: $\gamma \in \Q$, 
  $x^0 \in \Q^n$, $A \in \Q^{m \times n}$,
  $c \in \Q^n$, $b \in \Q^m$, and $r \in \mathbb{N}$,
  where $(A, b , r)$ is an encoding of the set $\S$
  in~\eqref{eq:milp} and $(c, \S, x^0)$ are input data for
  problem~\eqref{eq:imilp}.
  \item OUTPUT: YES, if $\exists d\in\D(x^0)$ such that
  $\|c-d\|\leq\gamma$, i.e., $\K(\gamma) \cap \D(x^0) \not= \emptyset$,
  NO otherwise.
  \end{itemize}
\end{definition}
Similarly, we have the dual bound verification problem for inverse MILP.
\begin{definition} \textbf{Inverse MILP Dual Bound Verification Problem
(IMDVP)}:
  \begin{itemize}
  \item INPUT: $\gamma \in \Q$,
  $x^0 \in \Q^n$, $A \in \Q^{m \times n}$,
  $c \in \Q^n$, $b \in \Q^m$, and $r \in \mathbb{N}$,
  where $(A, b , r)$ is an encoding of the set $\S$
  in~\eqref{eq:milp} and $(c, \S, x^0)$ are input data for
  problem~\eqref{eq:imilp}.
  \item OUTPUT: YES, if $\|c-d\| \geq \gamma$ for all $d\in\D(x^0)$.
  Equivalently, $\inter\left(\K(\gamma)\right) \cap \D(x^0) = \emptyset$,
  NO otherwise.
  \end{itemize}
\end{definition}
Finally, we have the optimal value verification problem.
\begin{definition} \textbf{Inverse MILP Optimal Value Verification Problem
(IMOVP)}:
  \begin{itemize}
  \item INPUT: $\gamma \in \Q$,
  $x^0 \in \Q^n$, $A \in \Q^{m \times n}$,
  $c \in \Q^n$, $b \in \Q^m$, and $r \in \mathbb{N}$,
  where $(A, b , r)$ is an encoding of the set $\S$
  in~\eqref{eq:milp} and $(c, \S, x^0)$ are input data for
  problem~\eqref{eq:imilp}.
  \item OUTPUT: YES, if $\min _{d \in \K(y) \cap \D(x^0)} y = \gamma$,
  NO otherwise.
  \end{itemize}
\end{definition}
In the rest of the section, we formally establish the complexity class
membership of each of the above problems and in so doing, illustrate the
relationships of the above problem to each other and to their MILP
counterparts. We assume from here on that $\conv(\S^+)$ full-dimensional (and
that hence, $\D(x^0)$ is also full-dimensional) to simplify the exposition.

\myparagraph{Informal Discussion} Before presenting the formal proofs, which are
somewhat technical, we informally describe the relationship of the inverse
optimization problem to the MILP analogues of the above decision problems,
which are described in Section~\ref{sec:complexity-classes}. Suppose we are
given a value $\alpha$ and we wish to determine whether it is a primal bound,
a dual bound, or the exact optimal value of~\eqref{eq:milp} with objective
function vector $c \in \Q^n$. Roughly speaking, we can utilize an algorithm
for solving~\eqref{eq:imilp} to make the determination, as follows. We first
construct the target vector
\begin{equation*}
x^\alpha := \alpha\frac{c}{\|c\|_2 ^2},
\end{equation*}
which has an objective function value (in the forward problem) of $c^\top
x^\alpha = \alpha$ by construction. Now suppose we solve~\eqref{eq:imilp} with
$x^\alpha$ as the target vector and $c$ as the estimated objective function
coefficient. Note that $x^\alpha$ is not necessarily in $\S$. Solving this
inverse problem will yield one of two results.
\begin{enumerate}
\item The optimal value of the inverse problem is $0$. This means that
$x^\alpha \in \argmax_{x \in \S^+} c^\top x$, which immediately implies that
$\alpha \geq \max_{x \in \S} c^\top x$ and we have that $\alpha$ is a valid dual
bound for~\eqref{eq:milp}.  
\item The optimal value of the inverse problem is strictly positive. Then
there must be a point $x \in \S$ for which $c^\top x > c^\top x^\alpha
= \alpha$, which means that $\alpha < \max_{x \in \S} c^\top x$ and  $\alpha$
is a (strict) primal bound for~\eqref{eq:milp}.
\end{enumerate}
There are a number of challenges to be overcome in moving from this informal
argument to the formal reductions in the completeness proofs. One obstacle is
that the above argument does not precisely establish the status of $\alpha$ as
a bound because we cannot distinguish between when $\alpha$ is a strict dual
bound (and hence not a primal bound) and when $\alpha$ is the exact optimal
value. This can be overcome essentially by appealing to
Proposition~\ref{prop:encoding} to reformulate certain strict inequalities as
standard inequalities (see Lemma~\ref{lem:epsilon}). A second challenge is
that we have only described a reduction of a decision version of MILP to an
optimization version of the inverse problem and have failed to describe any
sort of certificate. The formal proofs provide reductions to decision versions
of the inverse problem along with the required short certificates. Despite the
additional required machinery, however, the principle at the core of these
proofs is the simple one we have just described.

\myparagraph{Formal Proofs} We now present the main results of the paper and
their formal proofs. Two lemmas that characterize precisely when a given value
$\gamma$ is either a dual bound or a strict dual bound, respectively,
for~\eqref{eq:imilp} are presented first. These lemmas are the key element
underlying the proofs that follow so we first explain the intuition. Recall
the instances of~\eqref{eq:imilp} illustrated earlier in
Figure~\ref{fig:sets}. The output to the IMPVP for these four instances would
be NO, NO, YES, and NO, respectively. Note that for the YES instance in
Figure~\ref{fig:sets-3}, $K(\gamma)\cap\D(x^0)$ is nonempty, as one would
expect, while for the NO instance, this intersection is empty. It would appear
that we are thus facing an existential question, so that the YES output would
be the easier of the two to verify by simply producing an element of
$K(\gamma)\cap\D(x^0)$.

As it turns out, the above reasoning, though intuitive, is incorrect. The key
observation we exploit is that whenever $K(\gamma)\cap\D(x^0)$ is empty,
$\conv(\S)\cap\inter\left(\K^*(\gamma)\right)$ is nonempty and vice versa.
Furthermore, there \emph{is} a short certificate for membership in
$\conv(\S)\cap\inter\left(\K^*(\gamma)\right)$. The characterization in terms
of $K(\gamma)\cap\D(x^0)$ can be seen as being in the dual space, whereas the
characterization in terms of $\conv(\S)\cap\inter\left(\K^*(\gamma)\right)$ is
in the primal space. It is highly unlikely that a short certificate for
membership in $\D(x^0)$ exists and the reason can be understood upon closer
examination. It is that we lack an explicit description of $\D(x^0)$ in terms
of \emph{generators} (a V-representation). We only have access to a partial
description of it in terms of valid inequalities (a partial H-representation).
Whereas any convex combination of a subset of generators of a polyhedral set
must be in the set, a point satisfying a subset of the valid inequalities may
not be in the set. Therefore, a partial H-representation is not sufficient for
constructing a certificate. Even if we were able to obtain a set of generators
algorithmically, we have no short certificate of the fact that they are in
fact generators. On the other hand, it is easy to check membership for the
points in $\S$ that generate $\conv(\S)$. This is fundamentally why verifying
a primal bound for~\eqref{eq:milp} is in $\NPcomplexity$, whereas verifying a
primal bound for~\eqref{eq:imilp} is in $\coNP$.

The equivalence of the two characterizations described above is formalized
below and this is what eventually allows us to prove the existence of a short
certificate for $\gamma$ being a dual bound or strict dual bound
for~\eqref{eq:imilp}, respectively.
\begin{lemma}\label{lem:characterize-not-ub}
For $\gamma \in \Q$ such that $0 \leq \gamma < \|c\|$, we have
\begin{align*}
\K(\gamma) \cap \D(x^0) = \emptyset & \Leftrightarrow
\conv(\S^+) \cap \inter\left(\K^*(\gamma)\right) \not= \emptyset \\
& \Leftrightarrow \|c - d\| > \gamma \; \forall d \in \D(x^0) \\ 
& \Leftrightarrow \gamma \textrm{ is a strict dual bound
for~\eqref{eq:imilp}}\\ 
& \Leftrightarrow \gamma \textrm{ is \emph{not} a primal bound
for~\eqref{eq:imilp}}.
\end{align*}
\end{lemma} 

\begin{proof}
We prove that $\K(\gamma) \cap \D(x^0) = \emptyset$ if and only if
$\conv(\S^+) \cap \inter\left(\K^*(\gamma)\right) \not= \emptyset$. The
remaining implications follow by definition. 
\begin{itemize}
\item[($\Rightarrow$)] For the sake of contradiction, let us assume that both
$\K(\gamma) \cap \D(x^0) = \emptyset$ and
$\conv(\S^+) \cap \inter\left(\K^*(\gamma)\right) = \emptyset$. Under the
assumptions that $\conv(\S^+)$ is full-dimensional and $0 \leq \gamma
< \|c\|$, $\conv(\S^+)$ and $\inter\left(\K^*(\gamma)\right)$ are both
nonempty convex sets, so there exists a hyperplane separating them. In
particular, there exists $a \in \Re^n$ such that
\begin{equation}
\label{eq:sep} \tag{SEPi}
\max_{x \in \conv(\S^+)} a^\top x \leq \min_{x \in \K^*(\gamma)} a^\top x
= \inf_{x \in \inter\left(\K^*(\gamma)\right)} a^\top x. 
\end{equation}
The problem on the right-hand side is unbounded when
$a \not\in \K(\gamma)$, since then there must exist
$x \in \inter\left(\K^*(\gamma)\right)$ with $a^\top x < a^\top x^0$,
which means that $x - x^0$ is a ray with negative objective value
(recall $\K^*(\gamma)$ is a cone). Therefore, we must have
$a \in \K(\gamma)$ and it follows that $x^0$ is an optimal solution
for the problem on the right-hand side.
Therefore, we have
\begin{equation*}
\max_{x \in \conv(\S^+)} a^\top x \leq a^\top x^0.
\end{equation*}
Since $a \in \K(\gamma)$, then by assumption, $a \not\in \D(x^0)$, so
there exists an $\hat{x} \in \S^+$ such that $a^\top(x^0-\hat{x}) < 
0$. So finally, we have
\begin{equation*}
a^\top x^0 < a^\top \hat{x} \leq \max_{x \in \conv(\S^+)} a^\top x \leq a^\top
x^0,
\end{equation*}
which is a contradiction. This completes the proof of the forward direction.

\item[($\Leftarrow$)] For the reverse direction, we assume there exists
$\overline{x} \in \conv(\S^+) \cap \inter\left(\K^*(\gamma)\right)$. 
Since $\overline{x} \in \conv(\S^+)$, there exists
$\{x^1,x^2,\dots,x^k\} \subseteq \S^+$ and 
$\lambda \in \Q^k_+$ such that $\overline{x} = \sum
_{i=1} ^{k} \lambda_i x^i$, $ \sum _{i=1} ^{k} \lambda_i=1$, and $k \leq n+1$.
Now, let an arbitrary $d\in\K(\gamma)$ be given. Since $\gamma < \|c\|$, we
have $d \not=0$. Then, since
$\overline{x} \in \inter\left(\K^*(\gamma)\right)$, we have that
\begin{align*}
d^\top \left( x^0 - \overline{x} \right) < 0 & \Leftrightarrow
 d^\top x^0 - \left(\sum _{i=1} ^k \lambda_i x^i \right) < 0 \\
 & \Leftrightarrow
 d^\top \left(\sum _{i=1} ^k \lambda_i x^0 - \sum_{i=1}^k \lambda_i
 x^i \right) < 0 \\
 & \Leftrightarrow
 \sum _{i=1} ^k \lambda_i d^\top \left( x^0 -x^i \right) < 0\\
 & \Rightarrow
 \exists j \in \{1, \dots, k\} \textrm{ such that } d^\top (x^0 - x^j) < 0\\
 & \Rightarrow
 d \notin \D(x^0).
\end{align*}
Since $d$ was chosen arbitrarily, we have that $\K(\gamma)\cap\D(x^0)
=\emptyset$. This completes the proof of the reverse direction. 
\end{itemize}
\end{proof}

When $\gamma = 0$, we have that $\K^*(\gamma)$ is the half-space
$\{x \in \Re^n \mid c^\top x \geq c^\top x^0\}$, which further hints at the
relationship between the inverse dual bounding problem and both the MILP
primal bounding problem and the separation problem. A slightly modified
version of Lemma~\ref{lem:characterize-not-ub} characterizes when $\gamma$ is
a dual bound, but not necessarily strict. Note that unlike the
characterization in Lemma~\ref{lem:characterize-not-ub}, this characterization
doesn't hold when $\gamma = 0$, which is perhaps not surprising, since the
specific form of objective function we have chosen ensures that zero is always
a valid lower bound. As such, this information cannot be helpful in
distinguishing outcomes for the purpose of the reductions described shortly.

\begin{lemma} \label{lem:characterize-lb}
For $\gamma \in \Q$ such that $0 < \gamma < \|c\|$, we have
\begin{align*}
\inter\left(K(\gamma)\right) \cap \D(x^0) = \emptyset
& \Leftrightarrow \conv(\S^+) \cap \left(\K^*(\gamma) \setminus \{x^0\}\right)
\not= \emptyset \\ 
& \Leftrightarrow \|c - d\| \geq \gamma \;\forall d \in \D(x^0)  \\
& \Leftrightarrow \gamma \textrm{ is a dual bound 
for~\eqref{eq:imilp}} \\
& \Leftrightarrow \gamma \textrm{ is \emph{not} a strict primal bound 
for~\eqref{eq:imilp}} 
\end{align*}
\end{lemma} 

We now present the main theorems. 

\begin{theorem}
The IMPVP is in $\coNP$.
\label{th:conp}
\end{theorem}
\begin{proof}

We prove the theorem by showing the existence of a short certificate when the
output to the problem is NO. Therefore, let an instance $(\gamma,x^0,c,A,b,r)$
of the IMPVP for which the output is NO be given along with relevant input
data. Since the output is NO, we must have that $\gamma < \|c\|$, since $d
= 0$ is a valid solution otherwise.
By the characterization of Lemma~\ref{lem:characterize-not-ub}, the NO answer
is equivalent to the condition $\K(\gamma) \cap \D(x^0) = \emptyset$,
as well as to the existence of
$\overline{x} \in \conv(\S^+) \cap \inter\left(\K^*(\gamma)\right)$. We derive
our certificate from the latter, since this is an existence criterion. We have
noted earlier that simply providing such $\overline{x}$ is not itself a short
certificate, since we cannot verify membership in $\conv(\S^+)$ in polynomial
time. Fortunately, Carath\'eodory's Theorem provides that when
$\overline{x} \in \conv(\S^+)$, there exists a set of at most $n+1$ extreme
points of $\conv(S^+)$ whose convex combination yields $\overline{x}$.
Membership in $\inter\left(\K^*(\gamma)\right)$ is easily verified directly,
so the set of points serves as the certificate and is short, since only $n+1$
extreme points are needed.
\end{proof}
We next show that not only is IMPVP in $\coNP$, but it is also complete for it. 
\begin{theorem}
The IMPVP is complete for $\coNP$.
\label{th:conp-complete}
\end{theorem}
\begin{proof} We show that the MDVP can be Karp-reduced to the IMPVP. Let an
instance $(\alpha,c,A,b,r)$ of the MDVP be given. Then we claim this MDVP
can be decided by deciding an instance of the IMPVP with inputs $(0,
x^\alpha, c, A, b, r)$, where $x^\alpha = \alpha \frac{c}{\|c\|_2 ^2}$. By the
characterization of Lemma~\ref{lem:characterize-not-ub}, the IMPVP with this
input asks whether $\K(0) \cap \D(x^\alpha)$ is nonempty.
The first set contains a single point, $d=c$. The intersection is nonempty if
and only if $c$ is in $\D(x^\alpha)$. $c$ is in this cone
if and only if
\begin{equation*}
\begin{aligned}
c^\top \left(x - x^\alpha \right) \leq 0
&& \forall x \in \S & \Leftrightarrow
&c^\top x - \alpha &\leq 0 & \forall x \in \S,\\
&&& \Leftrightarrow &c^\top x &\leq \alpha & \forall x \in \S.
\end{aligned}
\end{equation*}
The last line above means the output of the MDVP is YES. This indicates that
the output of the original instance of the MDVP is YES if and only if the
output of the constructed instance of the IMPVP is YES.
\end{proof}
\Figurref{fig:conp} illustrates the reduction from the MDVP to the IMPVP for
three different $\alpha$ values. In the proof, the point $\alpha c$ is
constructed to have objective function value $\alpha$. The inverse
problem is then just to determine whether the inequality $c^\top
x \leq \alpha$ is valid for $\conv(\S)$, so the equivalence to the MDVP
follows straightforwardly. The output to both the MDVP and the IMPVP is NO
for $\alpha_1$ and YES for both $\alpha_2$ and $\alpha_3$. Note that this
proof crucially depends on the fact that we do not assume $x^0 \in \S$ for the
IMPVP.
\begin{figure}
\centering
\includegraphics[scale=1.0]{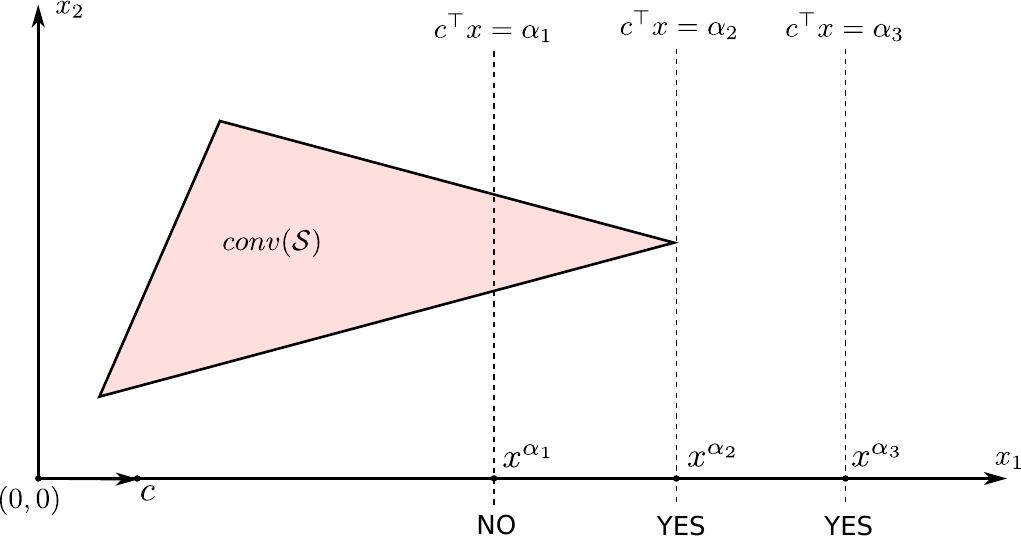}
\caption{Reduction Example \label{fig:conp}}
\end{figure}
Now, we consider IMDVP, for which the proofs take a similar approach. 
\begin{theorem}
The IMDVP is in $\NPcomplexity$.
\label{th:np}
\end{theorem}
\begin{proof}
The proof is almost identical to the proof that IMPVP is in $\coNP$. We show
that there exists a short certificate that can validate the output YES. Let an
instance $(\gamma,x^0,c,A,b,r)$ of the IMDVP be given such that the output is
YES. As in the proof of~\Theorref{th:conp}, it is enough to consider the case
where $\gamma < \|c\|$ and we can assume without loss of generality that
$\gamma > 0$ (otherwise, the output is trivially YES). When the output is YES,
by the characterization of Lemma~\ref{lem:characterize-lb},
$\inter\left(K(\gamma)\right)\cap\D(x^0)=\emptyset$ holds and there exists
$\overline{x} \in \conv(\S^+) \cap \K^*(\gamma)$ with 
$\overline{x} \not= x^0$. As in the proof
of~Theorem~\ref{th:conp}, Carath\'eodory's Theorem ensures there is a set of
at most $n+1$ members of $\S^+$ whose convex combination yields $\overline{x}$
and this set of points is the short certificate.
\end{proof}

To show that IMDVP is complete for $\NPcomplexity$, we need one more
technical lemma, which exploits Proposition~\ref{prop:encoding} to enable us
to Karp-reduce MPVP to IMDVP. The basic idea is to exploit the property that
the norm of the difference between two numbers can be bounded by a function of
their encoding lengths. This allows us to cleanly distinguish the cases
in which a given value is a primal bound for ~\eqref{eq:milp} from that in
which the given value is a strict dual bound by solving an IMDVP that we can
can easily construct.

\begin{lemma} \label{lem:epsilon}
Let $\alpha \in \Q$ and let $\epsilon := 2^{-\max\{\langle c \rangle
+ \nu, \langle \alpha \rangle\} - 1}$ and $\delta := 2^{-\langle x^0\rangle -
\nu - 1}$.
\begin{enumerate}[label=(\roman*)]
\item \label{lem:epsilon-1} If $\alpha \leq \max_{x \in \S} c^\top x +
\epsilon$, then $\alpha \leq \max_{x \in \S} c^\top x$.
\item  \label{lem:epsilon-2} 
If $\alpha \leq \max_{x \in \S} c^\top x$, then $\|c-d\| > \epsilon\delta > 0$ 
for all $d \in \D(x^{\alpha-\epsilon})$, where $x^{\alpha-\epsilon} =
(\alpha-\epsilon)\frac{c}{\|c\|^2_2}$.
\end{enumerate}
\end{lemma}

\begin{proof}
\begin{enumerate}[label=(\roman*)]

\item First, we have that the encoding length of $\max_{x \in \S} c^\top x$ is
bounded by $\langle c \rangle + \nu$. Then if $\alpha > \max_{x \in \S} c^\top
x$, $\alpha > \max_{x \in \S} c^\top x + \epsilon$ by
Proposition~\ref{prop:encoding}. Therefore, $\alpha \leq \max_{x \in \S}
c^\top x$.

\item Let $\alpha \leq \max_{x \in \S} c^\top x$ and
$\overline{x} \in \argmax_{x\in\S} c^\top x$ be given such 
that $\overline{x}$ is an extreme point of $\conv(\S)$ and let
$x^{\alpha-\epsilon} := (\alpha - \epsilon) \frac{c}{\|c\|^2_2}$. Since
$c^\top \overline{x} \geq \alpha$ and $c^\top x^{\alpha-\epsilon} = \alpha - \epsilon
< \alpha$, it follows that $\overline{x} \not= x^{\alpha-\epsilon}$. For an arbitrary
$d\in\D(x^{\alpha-\epsilon})$, we have that 
\begin{align}
d^\top (x^{\alpha-\epsilon} - \overline{x}) \geq 0
& \Leftrightarrow d^\top (x^{\alpha-\epsilon} - \overline{x}) - c^\top
(x^{\alpha-\epsilon} - \overline{x}) + c^\top 
(x^{\alpha-\epsilon} - \overline{x}) \geq 0 \notag \\
& \Leftrightarrow (d-c)^\top (x^{\alpha-\epsilon} - \overline{x})
\geq c^\top \overline{x} - c^\top x^{\alpha-\epsilon} \notag \\
& \Leftrightarrow (d-c)^\top (x^{\alpha-\epsilon} - \overline{x})
\geq c^\top \overline{x} - \alpha + \epsilon \label{eq:insertx0}\\
& \Leftrightarrow (d-c)^\top (x^{\alpha-\epsilon} - \overline{x})
\geq \epsilon \label{eq:use-existence-of-x}\\
& \Leftrightarrow \|c-d\| \|x^{\alpha-\epsilon} - \overline{x}\| \geq \epsilon
\label{eq:cauchy-schwarz}\\
& \Leftrightarrow \|c-d\| \geq \frac{\epsilon}{\|x^{\alpha-\epsilon}
- \overline{x}\|} \label{eq:div-allowed} \\
& \Leftrightarrow \|c-d\| > \epsilon\delta \label{eq:what}
\end{align}
Equation~\eqref{eq:insertx0} follows by substituting
$(\alpha-\epsilon)\frac{c}{\|c\|_2 ^2}$ for $x^{\alpha-\epsilon}$;
\eqref{eq:use-existence-of-x} follows from the nonnegativity of
$c^\top\overline{x}-\alpha$; and~\eqref{eq:cauchy-schwarz} follows from the
Cauchy--Schwarz inequality. Equation~\eqref{eq:div-allowed} follows
from~\eqref{eq:cauchy-schwarz} because $\|\overline{x} - x^{\alpha-\epsilon}\| > 0$.
Equation~\eqref{eq:what} follows from the fact that ${\|x^{\alpha-\epsilon} - \overline{x}\|}
< 2^{\langle x^{\alpha-\epsilon} \rangle + \nu + 1}$, again by
Proposition~\ref{prop:encoding}, since the 
encoding length of $x^{\alpha-\epsilon} - \overline{x}$ is bounded by the vertex complexity
$\nu$ of $\conv(\S)$ and $\langle x^{\alpha-\epsilon} \rangle$.
\end{enumerate}
\end{proof}

\begin{theorem}
The IMDVP is complete for $\NPcomplexity$.
\label{th:np-complete}
\end{theorem}
\begin{proof}
We show that the MPVP can be Karp-reduced to the IMDVP. Therefore, let an
instance $(\alpha,c,A,b,r)$ of the MPVP be given and let $\epsilon$ and
$\delta$ be given as in Lemma~\ref{lem:epsilon}.
Then we claim that the MPVP can be resolved by deciding the IMDVP with
inputs $(\epsilon\delta, x^{\alpha-\epsilon}, A, b, r)$, where
$x^{\alpha-\epsilon} = (\alpha - \epsilon) \frac{c}{\|c\|^2_2}$. The
IMDVP asks whether the set
$\inter\left(\K(\epsilon\delta)\right) \cap
\D(x^{\alpha-\epsilon})$
is empty. Note that $c \in \inter \left(\K(\epsilon\delta)\right)$. If the
output to the IMDVP is YES, then
$\inter\left(\K(\epsilon\delta)\right) \cap
\D(x^{\alpha-\epsilon})$
is empty. This means $c \not\in \D(x^{\alpha-\epsilon})$,
i.e.,
\begin{equation*}
c \notin \left\{d\in\Re ^n \mid
d^\top\left(x - x^{\alpha-\epsilon}\right) \leq 0 \; \forall x \in \S\right\}.
\end{equation*}
This in turn means that there exists an $\overline{x} \in \S$ such that
\begin{align*}
c^\top \left(\overline{x} - x^{\alpha-\epsilon} \right) > 0
& \Leftrightarrow c^\top \overline{x} -\alpha + \epsilon > 0 \\
& \Leftrightarrow c^\top \overline{x} + \epsilon > \alpha \\
& \Rightarrow c^\top \overline{x} \geq \alpha.
\end{align*}
The last implication is by Lemma~\ref{lem:epsilon}, part~\ref{lem:epsilon-1}.
Hence, the output to the MPVP is YES.

When the output to the IMDVP is NO, there exists
$\hat{d}\in\inter\left(\K(\epsilon\delta)\right)\cap\D(x^{\alpha-\epsilon})$.
Then $\hat{d}\in\D(x^{\alpha-\epsilon})$ and $\|c-\hat{d}\|
\leq \epsilon\delta$, so by the contraposition of Lemma~\ref{lem:epsilon},
part~\ref{lem:epsilon-2}, there is no $\overline{x}\in\S$ such that
$c^\top \overline{x} \geq \alpha$, i.e., $c^\top x < \alpha$ for all $x\in\S$.
This means the output to the MPVP is NO.
\end{proof}

Figure \ref{fig:np1} illustrates the case
in which $c^\top x < \alpha_3$ holds for all $x \in \S$, so that the output of
MPVP is NO. In this case, there exists
$\hat{d}\in\inter\left(\K(\epsilon\delta)\right)\cap\D(x^{\alpha-\epsilon})$.
In fact, $c$ is itself
in $\inter\left(\K(\epsilon\delta)\right)\cap\D(x^{\alpha-\epsilon})$, which
means the optimal value of the associated IMILP is 0. The instance of IMDVP
specified in the reduction also has output NO, since we are
checking the validity of a nonzero dual bound. 

\begin{figure}[tb]
\begin{center}
\subfloat[\label{fig:np1}]{\includegraphics[width=0.5\textwidth]
{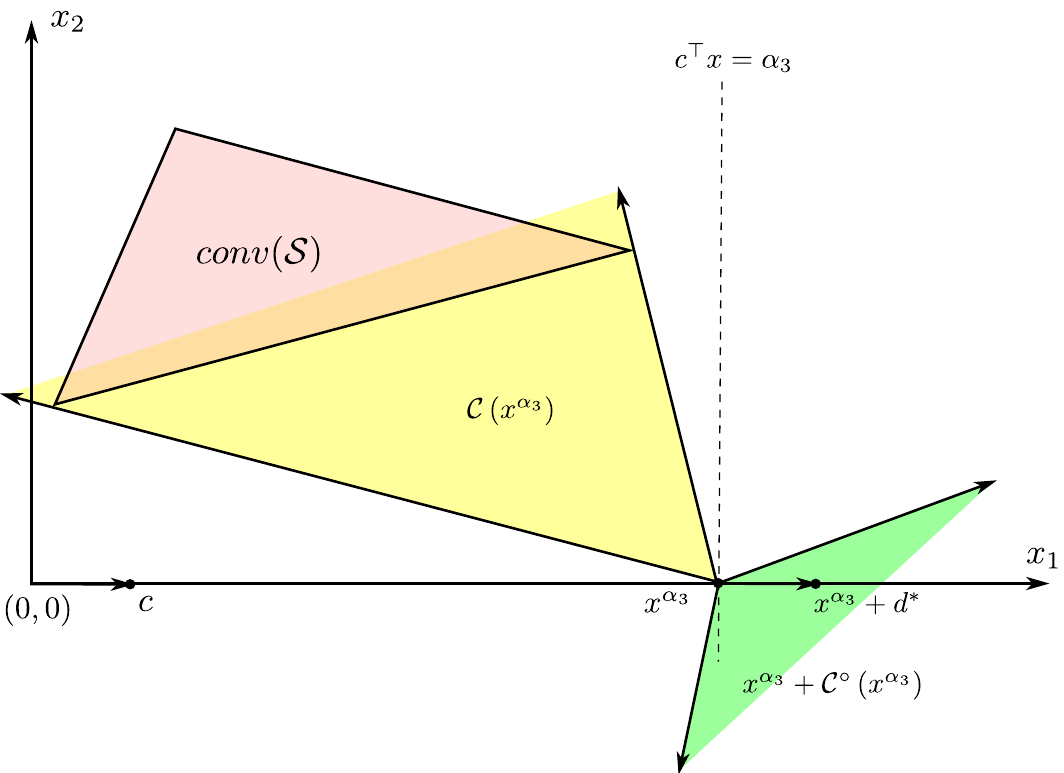}}
\subfloat[\label{fig:np2}]{\includegraphics[width=0.5\textwidth]
{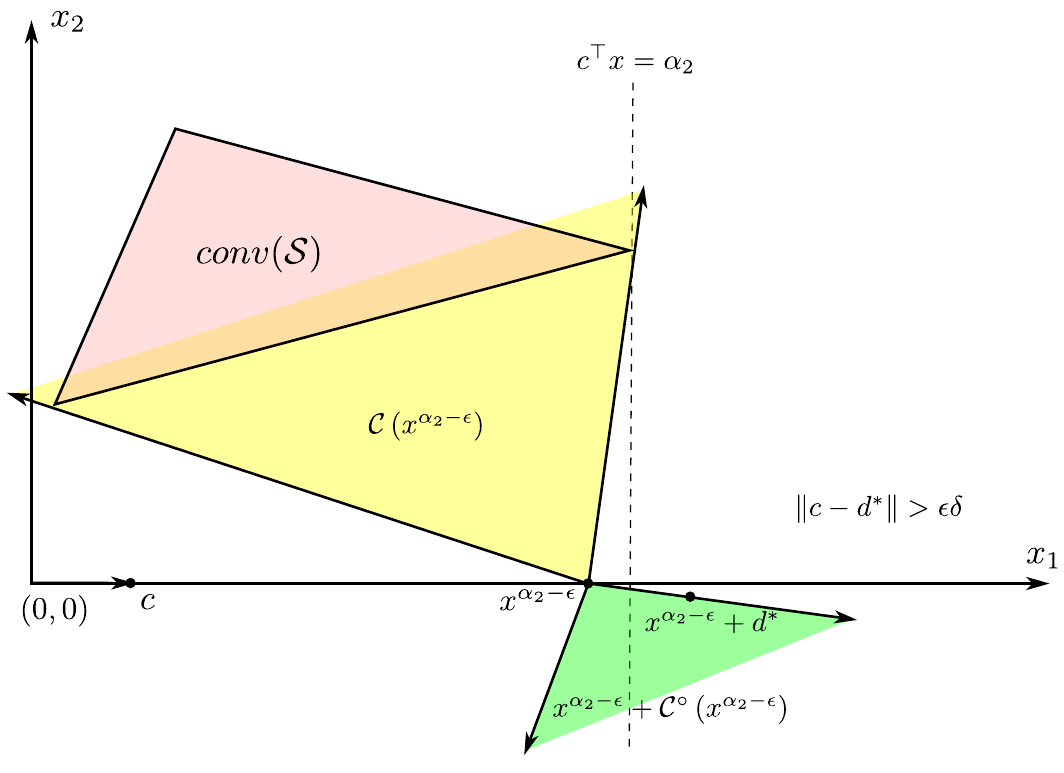}}
\end{center}
\caption{Lemma~\ref{lem:epsilon} on two simple examples \label{fig:np}}
\end{figure}

Figure \ref{fig:np2} illustrates the case in which the optimal value to the
MILP is exactly $\alpha_2$, so that the output of the MPVP is YES. In this
case, Lemma 3 tells us that the optimal value for the associated IMILP must be
greater than $\epsilon\delta$ and that the instance of IMDVP that we solve in
the reduction must have output YES as well. Note the essential role of
$\epsilon\delta$ as a number strictly between zero and the smallest possible
positive optimal value that the associated instance of IMILP can have. This is
necessary precisely because Lemma 2 does not hold for $\gamma = 0$. 

\begin{theorem}
The IMOVP is complete for $\Dcomplexity ^{\Pcomplexity}$.
\label{th:dp}
\end{theorem}
\begin{proof}
As noted before, the reduction presented in Theorem~\ref{th:np-complete} can
be used to reduce both the MDVP and the MPVP to the IMPVP and the IMDVP,
respectively. Using this reduction, the language of the IMOVP can then be
expressed as the intersection of the languages of the IMPVP and the IMDVP
that are in $\coNP$ and $\NPcomplexity$, respectively. This proves that the
IMOVP is in class $\Dcomplexity ^{\Pcomplexity}$. The IMOVP is complete for
$\Dcomplexity^{\Pcomplexity}$, since the MOVP can be reduced to the IMOVP using the
same reduction.
\end{proof}

Note that the optimal value verification problems associated with both the
inverse and forward problems are complete for $\Dcomplexity^{\Pcomplexity}$,
placing this decision version of both the forward and inverse problems in the
same complexity class.

\section{Conclusion and Future Directions}

In this paper, we discussed the relationship of the inverse mixed integer
linear optimization problem to both the associated forward and separation
problems.
We showed that the inverse problem can be seen as an optimization problem over
the cone described by all inequalities valid for $\S$ and satisfied at
equality by $x^0$ (the normal cone at $x^0$). Alternatively, it can also be
seen as an optimization problem over the 1-polar under some additional
assumptions. Both these characterizations make the connection with the
separation problem associated with $\S$ evident.

By observing that the separation problem for the feasible region of the
inverse problem under the $\ell_1$ and $\ell_{\infty}$ norms is an instance of
the forward problem, it can be shown via a straightforward cutting-plane
algorithm that the inverse problem can be solved by solving polynomially many
instances of the forward problem with different objective functions. This is
done by invoking the result of~\mycite{Gr\"otchel
et al.}{GroetschelLovaszSchrijver1993}, which shows that optimization and
separation are polynomially equivalent, i.e., each can be Cook-reduced to the
other. This in turn places the decision version of inverse MILP in the
complexity class $\Delta_2^p$.

The main result of the paper is that a stronger analysis is possible. We first
show that verification of a primal and dual bound for the inverse problem is in
$\coNP$ and $\NPcomplexity$ by providing short certificates for the NO and YES
outputs, respectively. We show that verification of a dual bound for the forward
problem can be Karp-reduced to verification of a primal bound for the
associated inverse problem. Thus, both problems are complete for the class
$\coNP$ and are on the same level of the polynomial-time hierarchy. Similarly,
verification of a primal bound for the forward problem can be Karp-reduced to
verification of a dual bound for the associated inverse problem. Thus, both of
those problems are complete for the class $\NPcomplexity$ and are also on the
same level of the polynomial-time hierarchy. Finally, we use these two results
to show that the verification problem for the optimal value of IMILP is
complete for the class $\Dcomplexity^{\Pcomplexity}$, which is the same class
that~\mycite{Papadimtriou}{Papadimitriou:1982:CF:800070.802199} showed
contains the MILP optimal value verification problem.

Although we have not done so formally, we believe the results in this paper
lay the groundwork for stating some version of Theorem~\ref{th:gls} that
incorporates the equivalence of the inverse problem and is stated in terms of
Karp reduction. The form such a result would take is not entirely obvious.
Whereas the essence of the separation problem is to determine whether or not a
given point is a member of a given convex set, the inverse problem implicitly
demands that we determine which of three sets contains a given point: the
relative interior of a given convex set, the boundary of that convex set, or
the complement of the set. Whereas it is known that the problem of determining
whether or not a given point is contained in the convex hull of solutions is
complete for $\NPcomplexity$, our results indicate that the related problem of
determining if a given point is on the boundary of the feasible set of an MILP
is a complete problem for $\Dcomplexity^{\Pcomplexity}$, while the problems of
determining whether a given point is contained in the convex hull of solutions
or determining whether a given point is contained in the complement of the
relative interior are in $\NPcomplexity$ and $\coNP$, respectively. Given all
of this, it seems clear that a unified result integrating all of the various
problems we have introduced and stating their equivalence in terms of Karp
reduction should be possible.

Finally, while we have implemented the cutting-plane algorithm in this paper,
it is clear that more work must be done to develop computationally efficient
algorithms for solving the inverse versions of difficult combinatorial
optimization problems. Development of customized algorithms for which the
number of oracle calls required in practice can be reduced is a next natural
step. Given our results, the existence of a direct algorithm for solving
inverse problems can also not be ruled out.

\bibliographystyle{plainnat}
\bibliography{complexity.bib}

\end{document}